\documentclass[aos,preprint]{imsart}

\usepackage{amsthm, amsmath, amssymb, graphicx, hyperref, tikz}
\usetikzlibrary{shapes,arrows,positioning}
\usepackage{verbatim}
\usepackage{natbib}
\usepackage{algpseudocode, subcaption}
\usepackage{algorithm}
\usepackage{arydshln} 
\usepackage{enumitem} 
\usepackage{subcaption,placeins}
\usepackage[titletoc]{appendix}

\usepackage[left=3.5cm,right=3.5cm,top=3.5cm,bottom=3.5cm]{geometry}

\makeatletter
\newtheorem*{rep@theorem}{\rep@title}
\newcommand{\newreptheorem}[2]{%
\newenvironment{rep#1}[1]{%
 \def\rep@title{#2 \ref{##1}}%
 \begin{rep@theorem}}%
 {\end{rep@theorem}}}
\makeatother

\newreptheorem{theorem}{Theorem}
\newreptheorem{lemma}{Lemma}
\newreptheorem{corollary}{Corollary}
\newreptheorem{example}{Example}
\newreptheorem{proposition}{Proposition}
\newtheorem{theorem}{Theorem}
\newtheorem{lemma}{Lemma}
\newtheorem{proposition}{Proposition}
\newtheorem{corollary}{Corollary}

\theoremstyle{remark}
\newtheorem{example}{Example}
\newtheorem{remark}{Remark}

\DeclareMathOperator{\rank}{rank}
\DeclareMathOperator{\linspan}{span}

\newcommand{\ignore}[1]{}

\newcommand{\bi}{\leftrightarrow}

\newcommand{\sign}{ \text{sgn} }
\numberwithin{equation}{section}

\newcommand{\pa}{{\rm pa}}       
\newcommand{\sib}{{\rm sib}}	   
\newcommand{\cardY}{V} 

\begin{document}

\begin{frontmatter}

\title{Computation of Maximum Likelihood Estimates in Cyclic Structural Equation Models}
\runtitle{Maximum likelihood estimation for structural equation models}
  
\author{\fnms{Mathias} \snm{Drton}\thanksref{m1}
  \ead[label=e3]{md5@uw.edu}}, 
  \author{\fnms{Christopher} \snm{Fox}\thanksref{m2}\ead[label=e1]{chrisfox.galton@gmail.com}}
    \and
  \author{\fnms{Y. Samuel} \snm{Wang}\thanksref{m1}\ead[label=e2]{ysamwang@uw.edu}}

\address{Department of Statistics\\
University of Washington\\
Seattle,   WA, U.S.A.\\
\printead{e2,e3}\\}
\address{Department of Statistics\\
The University of Chicago\\
Chicago, IL, U.S.A.\\
\printead{e1}
}

  \affiliation{ University of Washington \thanksmark{m1} and University of Chicago \thanksmark{m2}}

\begin{abstract}
  Software for computation of maximum likelihood estimates in linear
  structural equation models typically employs general techniques from
  non-linear optimization, such as quasi-Newton methods.  In practice,
  careful tuning of initial values is often required to avoid
  convergence issues.  As an alternative approach, we propose a
  block-coordinate descent method that cycles through the considered
  variables, updating only the parameters related to a given variable in
  each step.  We show that the resulting 
  block update problems can be solved in closed form even when the structural
  equation model comprises feedback cycles.  Furthermore, we give a
  characterization of the models for which the block-coordinate
  descent algorithm is well-defined, meaning that for generic data and
  starting values all block optimization problems admit a unique
  solution.  For the characterization, we represent each model by its
  mixed graph (also known as path diagram), which leads to
  criteria that can be checked in time that is polynomial in the
  number of considered variables.
\end{abstract}


\end{frontmatter}

\section{Introduction}
\label{sec:introBCD}

Structural equation models (SEMs) provide a general framework for
modeling stochastic dependence that arises through cause-effect
relationships between random variables.  The models form a cornerstone
of multivariate statistics with applications ranging from biology to
the social sciences \citep{bollen:1989,hoyle,kline}.  Through their
representation by path diagrams, which originate in the work of
\cite{wright:1921,wright:1934}, the models encompass directed
graphical models \citep{lauritzen:1996}.  While SEMs can naturally be
interpreted as models of causality that predict effects of
experimental interventions \citep{spirtes:2000,pearl:2009}, the focus
of this paper is on observational scenarios.  In other words, we
consider statistical inference based on a single independent sample
from a distribution in an SEM.  Concretely, we will treat linear SEMs
in which the effects of any latent variables are marginalized out and
represented through correlation among the error terms in the
structural equations; see e.g.~\citet[Section 3.7]{pearl:2009},
\citet[Chap.~6]{spirtes:2000} or \cite{wermuth:2011}.  This setting
arises, in particular, in problems of network recovery through model
selection as treated, e.g., by \cite{colombo:2012}, \cite{silva:2013}
or \cite{nowzohour:2015}.  For further details and references, see
Section 5.2 in \cite{drton:maathuis:2017}.

The specific problem we address is the computation of maximum
likelihood estimates (MLEs) in linear SEMs with Gaussian errors in the
structural equations.  The R packages `sem' \citep{fox:2006} and
`lavaan' \citep{lavaan} as well as commercial software
\citep{narayanan:2012} solve this problem by applying general
quasi-Newton methods for non-linear optimization.  However, these methods are often subject to convergence problems and may require careful choice of starting values  \citep{steiger:driving:2001}.  This is particularly exacerbated when computing MLEs in poorly fitting models as part of model
selection \citep{drton:eichler:richardson:2009}.  As a software manual puts it: ``It can be devilishly difficult for
software to obtain results for SEMs'' \cite[p.~112]{stata:2013}.

As an alternative, we propose a block-coordinate descent (BCD) method
that cycles through the considered variables, updating the parameters
related to a given variable in each step.  Each update is performed
through partial maximization of the likelihood function.  This method
generalizes the iterative conditional fitting algorithm of
\cite{chaudhuri:drton:richardson:2007} as well as the algorithm of
\cite{drton:eichler:richardson:2009}.  In contrast to this earlier
work, our extension is applicable to models that comprise feedback
cycles.  Models with feedback cycles have been treated by \cite{spir:uaicyc},
\cite{Richardson96Algorithm,rich:dcgmarkov}, and more
recently by \cite{Lacerda08}, \cite{mooij:heskes:2013} and
\cite{raskutti:2016}.  An example of a recent application can be found
in the work of \cite{nature}.

The presence of feedback loops complicates likelihood inference as
even in settings without latent variables MLEs are generally
high-degree algebraic functions of the data.  For example, the MLE in
the model given by the graph in Figure~\ref{fig:mldegree} is  an
algebraic function of degree 7; see Chapter 2.1 in \cite{oberwolfach}
for how to compute this ML degree.  Somewhat surprisingly, however, the
update steps in our BCD algorithm admit a closed form even in the
presence of feedback loops, and the computational effort is on the
same order as in the case without feedback loops.  In numerical
experiments the BCD algorithm is seen to avoid convergence problems.

\begin{figure}[t]
  \centering
      \begin{tikzpicture}[->,>=triangle 45,shorten >=1pt,
        auto,thick,
        main node/.style={circle,inner
          sep=2pt,fill=gray!20,draw,font=\sffamily}]  
      
        \node[main node] (1) {1}; 
        \node[main node] (2) [right=1.5cm of 1] {2}; 
        \node[main node] (3) [right=1.5cm of 2] {3}; 
        \node[main node] (4) [right=1.5cm of 3] {4};
        \node[main node] (5) [right=1.5cm of 4] {5};
      
        \path[color=black!20!blue,every
        node/.style={font=\sffamily\small}] 
        (1) edge node {} (2)
        (5) edge node {} (4);
        \path[color=black!20!blue,every
        node/.style={font=\sffamily\small}, bend right] 
        (2) edge node {} (3)
        (3) edge node {} (4)
        (4) edge node {} (3)
        (3) edge node {} (2)
        ;
      \end{tikzpicture}
      \caption{Graph of a cyclic linear SEM with maximum likelihood
        degree 7.}
  \label{fig:mldegree}
\end{figure}
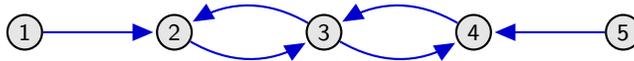

As a second main contribution, we show that the algorithm applies to
interesting models with `bows'.  In terms of the mixed graph/path
diagram, a bow is a subgraph on two nodes $i$ and $j$ with two edges
$i\to j$ and $i\bi j$.  Such a subgraph indicates that there is both a
direct effect of the $i$-th variable on the $j$-th variable as well as a latent
confounder with effects on the two variables.  Bows can lead to
collinearity issues in the BCD algorithm, and we are able to give a
characterization of the models for which the algorithm is
well-defined, meaning that for generic data and starting values all
block optimization problems admit a unique and feasible solution.  
For the characterization, we represent
each model by its mixed graph/path diagram, which leads to criteria
that can be checked in time that is polynomial in the number of
considered variables.

The paper is organized as follows.  In Section \ref{sec:sem}, we
review necessary background on linear SEMs.
The new BCD
algorithm is derived in Section \ref{sec:cycles}.  Its properties are
discussed in Section \ref{sec:assumptionDiscussion}.  Numerical
examples are presented in Section \ref{sec:simulations}.  Finally, we
conclude with a discussion of the considered problem in Section
\ref{sec:conclusion}.


\section{Linear structural equation models}
\label{sec:sem}

\subsection{Basics}

A structural equation model (SEM) captures dependence among a set of
variables $\{Y_i: i\in\cardY\}$.
Each model is built from a system of equations, with one equation for
each considered variable.  Each such \emph{structural} equation
specifies how a variable $Y_i$
arises as a function of the other variables and a stochastic error
term $\epsilon_i$.  In the linear case considered here, we have
\begin{equation}
Y_i = \sum_{j \in \cardY \setminus \{i\}} \beta_{ij}Y_j + \epsilon_i,
\qquad i\in \cardY.
\label{eq:structural_eqns}
\end{equation}
Collecting the $Y_i$
and $\epsilon_i$
terms into the vectors $Y$
and $\epsilon$,
respectively, (\ref{eq:structural_eqns}) can be rewritten as
\begin{equation} \label{eq:matrix_form_Y}
Y = BY + \epsilon,
\end{equation}
where $B = (\beta_{ij})$ is a matrix of coefficients that
are sometimes termed \emph{structural parameters} \citep{bollen:1989}.
Specific models of interest are obtained by assuming that for some
index pairs $(i,j)$, variable $Y_j$ has no direct effect on $Y_i$,
which in the linear framework is encoded by the restriction that
$\beta_{ij} = 0$.

Techniques for statistical inference are often based on the assumption
that $\epsilon$
follows a multivariate normal distribution with possible dependence
among its coordinates.  So,
\begin{equation} \label{eq:epsilon}
\epsilon \sim \mathcal{N}(0,\Omega),
\end{equation}
where $\Omega=(\omega_{ij})$ is a symmetric, positive definite matrix
of parameters.  An entry $\omega_{ij}$ may capture effects of
potential latent variables that are common causes of $Y_i$ and $Y_j$.
When no latent common cause of $Y_i$ and $Y_j$ is believed to exist,
constrain $\omega_{ij} = \omega_{ji} = 0$ \citep[see e.g.][]{spirtes:2000,pearl:2009}.  As a result of
(\ref{eq:matrix_form_Y}) and (\ref{eq:epsilon}), the observed random
variables, $Y$, have a centered normal distribution with covariance
matrix
\begin{equation} \label{eq:sigmaMap}
\Sigma = (I-B)^{-1}\Omega(I-B)^{-T}.
\end{equation}
Here, $I$ is the $V\times V$ identity matrix.  Note that the assumption of
centered variables can be made without loss of generality
\cite[Chapter 7]{MR1990662}.

It is often convenient to represent an SEM by a mixed graph or path
diagram \citep{wright:1921, wright:1934}.  The graph has vertex set
$V$ and is mixed in the sense of having both a set of \emph{directed
  edges} $E_{\to}$ and a set of \emph{bi-directed edges} $E_{\bi}$.
The directed edges in $E_{\to}$ are ordered pairs in $V\times V$,
whereas the edges in $E_{\bi}$ have no orientation and are unordered
pairs $\{i,j\}$ with $i,j\in V$.  We will often write $i\to j$ in
place of $(i,j)$ for a potential edge in $E_{\to}$ and $i\bi j$ for a
potential edge $\{i,j\}$ in $E_{\bi}$.  In this setup, each variable
$Y_i$ is then represented by a node, corresponding to its index
$i \in V$.  An edge $j\to i$ is not in $E_{\to}$ if and only if the
model imposes the constraint that $\beta_{ij} = 0$.  Note that in our
context there are no self-loops $i\to i$.  Similarly, the edge
$i\bi j$ is absent from $E_{\bi}$ if and only if the model imposes the
constraint that $\omega_{ij} = \omega_{ji} = 0$.  Finally, for each
node $j\in V$, we define two sets $\pa(j)$ and $\sib(j)$ that we refer
to as the \emph{parents} and \emph{siblings} of $j$, respectively.
The set $\pa(j)$ comprises all nodes $i\in V$ such that
$i\to j \in E_{\to}$, and $\sib(j)$ is the set of all nodes $i\in V$
such that $i\bi j \in E_{\bi}$.

Let $G=(V,E_{\to},E_{\bi})$ be a mixed graph, and define
$\mathbf{B}(G)$ to be the set of real $V\times V$ matrices
$B=\left(\beta_{ij}\right)$ such that $I-B$ is invertible and
\begin{equation}
\beta_{ij}=0 \quad \text{ whenever } j\to i \notin E_{\to}.
\label{eq:BofG}
\end{equation}
Similarly, define $\mathbf{\Omega}(G)$ to be the set of all positive
definite symmetric $V\times V$ matrices $\Omega = (\omega_{ij})$ that
satisfy
\begin{equation}
\omega_{ij} = 0 \quad \text{ whenever } j\bi i \notin E_{\bi}.
\label{eq:OofG}
\end{equation}
The \emph{linear SEM} $\mathbf{N}(G)$ 
associated with graph $G$ is then the family of multivariate
normal distributions $\mathcal{N}(0,\Sigma)$ with covariance matrix
$\Sigma$ as in~(\ref{eq:sigmaMap}) for $B\in\mathbf{B}(G)$ and
$\Omega\in\mathbf{\Omega}(G)$.

A mixed graph $G$ and the associated model $\mathbf{N}(G)$ are
\emph{cyclic} if $G$ contains a directed cycle, that is, a subgraph of
the form
\[
i_1 \to i_2 \to \cdots \to i_k \to i_1
\]
for distinct nodes $i_1,\ldots i_k \in V$, $k\ge 2$.  If there is no
such cycle, the graph and corresponding model are said to be
\emph{acyclic}.  Acyclicity brings about great simplifications as we
have $\det(I-B)=1$ for every $B\in\mathbf{B}(G)$ if and only if $G$ is
acyclic.  To see this note that when $G$ is acyclic, there exists a
topological ordering of $V$, i.e., a relabeling of $V$ such that
$i\to j \in E_{\to}$ only if $i < j$.  Under such an ordering every
matrix in $\mathbf{B}(G)$ is strictly lower triangular.  If $G$ is an
acyclic digraph, such that $E_{\bi}=\emptyset$ then the MLE in
$\mathbf{N}(G)$ is obtained by solving a linear regression problem for
each variable $Y_i$, $i\in V$.  For an acyclic graph with
$E_{\bi}\not=\emptyset$, this is generally no longer the case but the
MLE can be found by iterative least squares computations
\citep{drton:eichler:richardson:2009}.

\subsection{Cyclic models}

A challenge in the computation of MLEs in models with cyclic path
diagrams is the fact that $\det(I-B)$ is not constant one.
For example, $\det(I-B) = 1-\beta_{32}\beta_{43}\beta_{24}$
for matrices $B\in\mathbf{B}(G)$ when $G$ is the mixed graph in Figure
\ref{fig:example}.  We observe a correspondence between the term
$\beta_{32}\beta_{43}\beta_{24}$ and the directed cycle
$2\to 3 \to 4 \to 2$ in the graph.  We now review this connection in
the setting of a general mixed graph $G$.

Let $\mathbf{S}_V$ be the symmetric group of all permutations of the vertex set $V$.  Every permutation
$\sigma\in\mathbf{S}_V$ has a unique decomposition into disjoint
permutation cycles.  Let $\mathcal{C}(\sigma)$ be the set of
permutation cycles of $\sigma$, and let $\mathcal{C}_2(\sigma)$ be the
subset containing cycles of length 2 or more.  Write $n(\sigma)$ for
the cardinality of $\mathcal{C}_2(\sigma)$, and $V(\sigma)$ for the
set of nodes that are contained in a cycle in $\mathcal{C}_2(\sigma)$.
Moreover,
define
\begin{equation}
\mathbf{S}_V(G) = \{\sigma\in \mathbf{S}_V \ :\  i=\sigma(i) \
\text{ or } \ i\to\sigma(i) \in
E_{\to} \ \text{ for all } i\in V \}.
\end{equation}

\begin{lemma}
  \label{lem:det}
  Let $B=(\beta_{ij})\in\mathbf{B}(G)$
  for a mixed graph $G=(V,E_{\to},E_{\bi})$. Then
  \[
    \det{(I-B)} \ = \sum_{\sigma\in\mathbf{S}_V(G)}(-1)^{n(\sigma)}
    \prod_{i\in V(\sigma)}\beta_{\sigma(i),i}. 
  \]
\end{lemma}

The lemma follows from a Leibniz expansion of the determinant.  It
could be derived from Theorem 1 in \cite{harary:1962} by treating the
diagonal of $I-B$ as self-loops with weight 1 and taking into account
that $B$ is negated.  Because the lemma is of importance for later
developments, we include its proof in Appendix~\ref{app:lem:det}.

When deriving the block-coordinate descent algorithm proposed in
Section~\ref{sec:cycles}, we treat $\det(I-B)$ as a function of only
the entries in a given row.  By multilinearity of the determinant this
function is linear and its coefficients are obtained in a Laplace
expansion.  Throughout the paper, we let $-i:=V\setminus\{i\}$ and
denote the $U\times W$ submatrix of a matrix $A$ by $A_{U,W}$.

\begin{figure}[t]
  \centering
      \begin{tikzpicture}[->,>=triangle 45,shorten >=1pt,
        auto,thick,
        main node/.style={circle,inner
          sep=2pt,fill=gray!20,draw,font=\sffamily}]  
      
        \node[main node] (1) {1}; 
        \node[main node] (2) [right=1.5cm of 1] {2}; 
        \node[main node] (3) [right=1.5cm of 2] {3}; 
        \node[main node] (4) [below=1.5cm of 3] {4};
        \node[main node] (5) [left=1.5cm of 4] {5};
      	\node[main node] (6) [left=1.5cm of 5] {6};
      
        \path[color=black!20!blue,every
        node/.style={font=\sffamily\small}] 
        (1) edge node {} (2)
        (2) edge node {} (3)
        (3) edge node {} (4)
        (4) edge node {} (5)
        (4) edge node {} (2)
        (5) edge node {} (6);
        
        \path[color=black!20!blue,every
        node/.style={font=\sffamily\small}, style={<->, dashed}] 

        (2) edge node {} (5)
        (3) edge node {} (5)
        ;
      \end{tikzpicture}
   \caption{\label{fig:example}Cyclic mixed graph that is almost everywhere identifiable.}

\end{figure}
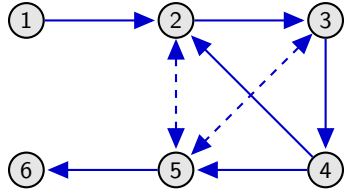

\begin{lemma}
  \label{thm:detForm}
  Let $B=(\beta_{ij})\in\mathbf{B}(G)$ for a mixed graph
  $G=(V,E_{\to},E_{\bi})$. Fix an arbitrary node $i\in V$.  Then
  $\det(I-B)$ is linear in the entries of
  $B_{i,\pa(i)}=(\beta_{ij}:j\in\pa(i))$ with
  \[
    \det(I-B) = c_{i,0} + B_{i,\pa(i)} c_{i,\pa(i)},
  \]
  where $c_{i,0}\in\mathbb{R}$ and the entries of
  $c_{i,\pa(i)}\in\mathbb{R}^{\pa(i)}$ are subdeterminants, namely, 
  \[
    c_{i,0} \;=\; \det\left( (I-B)_{-i,-i} \right), \qquad
    c_{i,p} \;=\; (-1)^{i+p-1}\det\left( (I-B)_{-i,-p}\right), \quad p\in\pa(i);
  \]
  to define $(-1)^{i+p-1}$ enumerate $V$ in accordance with
  the layout of the matrix $I-B$.
\end{lemma}

\begin{example} \rm
  The mixed graph $G$ from Figure \ref{fig:example} encodes the
  equation system
\[
\begin{array}{lcl} 
Y_1 = \epsilon_1, & \hspace{4em} & Y_2 = \beta_{21}Y_1+\beta_{24}Y_4+\epsilon_2, \\
Y_3 = \beta_{32}Y_2+\epsilon_3, & \hspace{4em} & Y_4 = \beta_{43}Y_3+\epsilon_4, \\
Y_5 = \beta_{54}Y_4+\epsilon_5, & \hspace{4em} & Y_6 = \beta_{65}Y_5+\epsilon_6, \\
\end{array}
\]
where $\epsilon_1$, $\epsilon_2$, $\epsilon_3$, $\epsilon_4$, and
$\epsilon_6$ are all pairwise uncorrelated, and $\epsilon_5$ is
uncorrelated with $\epsilon_1$, $\epsilon_4$, and $\epsilon_6$.  The system contains the directed cycle $2\to3 \to 4 \to2$. Consequently, 
\[\det(I-B) = 1 - \beta_{32}\beta_{43}\beta_{24}.
\]
Hence, the coefficients must satisfy
$\beta_{32}\beta_{43}\beta_{24}\neq 1$ for the equation system to
yield a positive definite covariance matrix.  When fixing node
$i\in V$ and writing $\det(I-B)$ as a linear function of
$(\beta_{ij})_{j\in\pa(i)}$ as in Lemma~\ref{thm:detForm}, we have
\[
\begin{array}{lll}
  c_{1,0} = 1-\beta_{32}\beta_{43}\beta_{24}, & \pa(1) = \emptyset; & \\
  c_{2,0} = 1, & \pa(2) = \{1,4\}, & c_{2,\pa(2)} = \left(0,-\beta_{32}\beta_{43}\right)^T; \\
  c_{3,0} = 1, & \pa(3) = \{2\}, & c_{3,\pa(3)} = -\beta_{43}\beta_{24}; \\
   c_{4,0} = 1, & \pa(4) = \{3\}, & c_{4,\pa(4)} = -\beta_{32}\beta_{24};  \\
   c_{5,0} = 1-\beta_{32}\beta_{43}\beta_{24}, & \pa(5) = \{4\}, & c_{5,\pa(5)} = 0;  \\
   c_{6,0} = 1-\beta_{32}\beta_{43}\beta_{24}, & \pa(6) = \{5\}, & c_{6,\pa(6)} = 0.  \\
\end{array}
\]
\end{example}


\subsection{Likelihood inference}
\label{sec:lik_inference}

Suppose we are given a sample of $N$ observations in $\mathbb{R}^V$.
Let $Y$ be the $\cardY \times N$ matrix with these observations as
columns, and let $S = \frac{1}{N}YY^T$ be the associated $V\times V$
sample covariance matrix (for known zero mean).  Fix a possibly cyclic
mixed graph $G$.  Ignoring an additive constant and dividing out a
factor of $N/2$, model $\mathbf{N}(G)$ has log-likelihood function
\begin{align} 
  \notag
  \ell_{G,Y}(\Omega,B) 
  &=
  -\log\det\left((I-B)^{-1}\Omega (I-B)^{-T}\right) -
  \text{tr}\left\{(I-B)^T\Omega^{-1}(I-B)S\right\}
  \\
\label{eq:log_likelihood} 
&=
  -\log\det(\Omega)-\log\det(I-B)^2-\text{tr}\left\{(I-B)^T\Omega^{-1}(I-B)S\right\}. 
\end{align}

Throughout the paper, we assume that $Y$ has full rank $|V|$.  This
holds with probability one if the sample is from a continuous
distribution and $N\geq |\cardY|$.  Full rank of $Y$ implies that $S$
is positive definite, and the log-likelihood function $\ell_{G,Y}$ is
then bounded for any graph $G$.  However, if $G$ is sparse with a
bi-directed part $(V,E_{\bi})$ that is not connected, then
$\ell_{G,Y}$ may also be bounded if $S$ is not positive definite
\citep{fox:2014}.

Our problem of interest is to compute (local) maxima of the
log-likelihood function.  These solve the likelihood equations, which
are obtained by equating to zero the gradient of
$\ell_{G,Y}(B,\Omega)$.  To be precise, the partial derivatives are
taken with respect to the free entries in $B$ and $\Omega$, which we
denote by $\beta$ and $\omega$, respectively.  So, $\beta$ has
$|E_{\to}|$ entries, and $\omega$ has $|V|+|E_{\bi}|$ entries.  Let
$\text{vec}(A)$ denote the vectorization (stacking of the columns) of
a matrix $A$.  Then there are 0/1-valued matrices $P$ and $Q$ such
that $\text{vec}(B) = P\beta$, and $\text{vec}(\Omega)=Q\omega$.

\begin{proposition}
The likelihood equations of the model $\mathbf{N}(G)$ can be written as
\begin{align} \label{eq:likelihood_eq1}
P^T \ \mathrm{vec}\left[\Omega^{-1}(I-B)S -(I-B)^T \right] &= 0,
\\
 \label{eq:likelihood_eq2}
Q^T \ \mathrm{vec}\left(\Omega^{-1}-\Omega^{-1}(I-B)S(I-B)^T\Omega^{-1}\right) &= 0.
\end{align}

\end{proposition}
A derivation of this result is provided in Appendix
\ref{chap:lik_eqns}.  In general, the likelihood equations are
difficult to solve analytically; recall the example from
Figure~\ref{fig:mldegree}.  Instead, it is common practice to use
iterative maximization techniques. 

\section{Block-coordinate descent for cyclic mixed graphs}
\label{sec:cycles}

\subsection{Algorithm overview}

We now introduce our block-coordinate descent (BCD) procedure for
computing the MLE in a possibly cyclic mixed graph model
$\mathbf{N}(G)$.
The method requires initializing with a choice of $B\in\mathbf{B}(G)$
and $\Omega\in\mathbf{\Omega}(G)$.  The algorithm then proceeds by
repeatedly iterating through all nodes in $V$ and performing update
steps.  In the update for node $i$, we maximize the log-likelihood
function with respect to all parameters corresponding to edges with a
head at $i$ (i.e., $B_{i,\pa(i)}$ and $\Omega_{i,\sib(i)\cup\{i\}}$)
while holding all other structural parameters fixed.  The parameters
that are updated determine the $i$-th row in $B$ and the $i$-th row
and column in the symmetric matrix $\Omega$.  The algorithm stops when
a convergence criterion is satisfied.

In the derivation of the block update, we write $Y_C$ for the
$C\times N$ submatrix of $Y$, for subset $C\subset V$.  In particular,
$Y_{-i}= Y_{V\setminus \{i\}}$ and 
$Y_i$ is the $i$-th row of $Y$.
 Finally, we note that we will invoke assumptions to ensure
that the optimization problem yielding the block update admits a
unique solution.  The graphs $G$ for which these assumptions hold will
be characterized in Section \ref{sec:assumptionDiscussion}.

\subsection{Block update problem}
\label{sec:block-updates}

In the $i$-th block update problem, we seek to maximize the
log-likelihood function $\ell_{G,Y}$ while holding the submatrices
$\Omega_{-i,-i}$ and $B_{-i}$ fixed.  Let
\[
\omega_{ii.-i} =
\omega_{ii}-\Omega_{i,-i}\Omega_{-i,-i}^{-1}\Omega_{-i,i}
\]
 be the
conditional variance of the error term $\epsilon_i$ given
$\epsilon_{-i}$; here $\Omega_{-i,-i}^{-1}= (\Omega_{-i,-i})^{-1}$.
In analogy to Theorem 12 in \citet{drton:eichler:richardson:2009}, the
log-likelihood function can be decomposed as
\begin{align}
  \ell_{G,Y}(\Omega,B) 
  &=  
 -\log{\omega_{ii.-i}}-\frac{1}{N\omega_{ii.-i}}\|Y_i-B_{i,\pa(i)}Y_{\pa(i)}-\Omega_{i,\sib(i)}(\Omega_{-i,-i}^{-1}\epsilon_{-i})_{\sib(i)}\|^2 \notag  \\
 &\hspace{2em}  -\log\det(\Omega_{-i,-i})-\frac{1}{2N}\text{tr}(\Omega_{-i,-i}^{-1}\epsilon_{-i}\epsilon_{-i}^T) + \log\det(I-B)^2. \label{eq:decomp}
\end{align}
This follows by factoring the joint distribution of $\epsilon$
into the marginal distribution of $\epsilon_{-i}$ and the conditional
distribution of $\epsilon_i$ given $\epsilon_{-i}$.  The key difference
between~(\ref{eq:decomp}) and the corresponding log-likelihood
decomposition in \citet{drton:eichler:richardson:2009} is the presence
of the term $\log\det(I-B)^2$, which is nonzero for cyclic graphs.

With $\Omega_{-i,-i}$ and $B_{-i}$ fixed, we can first compute the
error terms
\begin{equation} \label{eq:fixed_residuals}
\epsilon_{-i} = (I-B)_{-i}Y
\end{equation}
and subsequently the \emph{pseudo-variables}
\begin{equation} \label{eq:pseudo_variables}
Z_{-i} = \Omega_{-i,-i}^{-1}\epsilon_{-i}.
\end{equation}
From (\ref{eq:decomp}), it is clear that, for fixed $\Omega_{-i,-i}$
and $B_{-i}$, the maximization of $\ell_{G,Y}$ reduces to the maximization
of the function
\begin{multline} \label{eq:max}
\ell_{G,Y,i}\left(\Omega_{i,\sib(i)},\omega_{ii.-i},B_{i,\pa(i)}\right) 
= -\log{\omega_{ii.-i}}-\frac{1}{N\omega_{ii.-i}}\|Y_i-B_{i,\pa(i)}Y_{\pa(i)}-\Omega_{i,\sib(i)}Z_{\sib(i)}\|^2\\
 \quad +\log[(c_{i,0} + B_{i,\pa(i)}c_{i,\pa(i)})^2].
\end{multline}
Here, we applied Lemma~\ref{thm:detForm}, and let
$B_{i,\pa(i)}=(\beta_{ij}:j\in\pa(i))$ and
$\Omega_{i,\sib(i)}=(\omega_{ik}:k\in \sib(i))$.  The domain of
definition of $\ell_{G,Y,i}$ is
$\mathbb{R}^{\sib(i)}\times(0,\infty)\times\mathbb{R}^{\pa(i)}_{\text{inv}}$,
where
\[
\mathbb{R}^{\pa(i)}_{\text{inv}} = \mathbb{R}^{\pa(i)}\setminus \{ B_{i,\pa(i)} :
c_{i,0}+B_{i,\pa(i)}c_{i,\pa(i)}=0\}
\]
 excludes choices of $B_{i,\pa(i)}$ for which $I-B$ is non-invertible.

For any fixed choice of $B_{i,\pa(i)}$ and $\Omega_{i,\sib(i)}$, if
$Y_i-B_{i,\pa(i)}Y_{\pa(i)}-\Omega_{i,\sib(i)}Z_{\sib(i)}
\neq 0$,   then
\begin{equation}
\omega^\star_{ii.-i} =
\frac{1}{N}\|Y_i-B_{i,\pa(i)}Y_{\pa(i)}-\Omega_{i,\sib(i)}Z_{\sib(i)}\|^2 
\label{eq:varestimate}
\end{equation}
uniquely maximizes $\ell_{G,Y,i}$ with respect to $\omega_{ii.-i}$.
This fact could be used to form a profile log-likelihood function.
Before proceeding, however, we shall address the concern that for a
mixed graph $G$ that contains cycles, it may occur that
$Y_i \in \linspan(Y_{\pa(i)}, Z_{\sib(i)})$ even if the rows of $Y$
are linearly independent.
A simple example would be the graph with nodes 1 and 2 and three edges
$1\to 2$, $1\leftarrow 2$ and $1\bi 2$; see
Example~\ref{ex:twonodes:threeedges} below.

\begin{lemma}
  \label{lem:wii.-i:profile}
  Let the data matrix $Y\in\mathbb{R}^{V\times N}$ have linearly independent
  rows.  Then
  $Y_i-B_{i,\pa(i)}Y_{\pa(i)}-\Omega_{i,\sib(i)}Z_{\sib(i)}\not=0$ for
  all $B\in\mathbf{B}(G)$, $\Omega\in\mathbf{\Omega}(G)$ and $i\in V$.
\end{lemma}
\begin{proof}
  From (\ref{eq:pseudo_variables}), 
  \[
    Y_i-B_{i,\pa(i)}Y_{\pa(i)}-\Omega_{i,\sib(i)}Z_{\sib(i)} =
  \epsilon_i - \omega_{i,-i}\Omega_{-i,-i}^{-1}\epsilon_{-i}=0
  \]
  only if $\epsilon=(I-B)Y\in\mathbb{R}^{V\times N}$ has linearly
  dependent rows.  However, this cannot occur when $Y$ has linearly
  independent rows as matrices $B\in\mathbf{B}(G)$ have $I-B$
  invertible.
\end{proof}

According to Lemma~\ref{lem:wii.-i:profile}, we may indeed
substitute $\omega^\star_{ii.-i}$ from~(\ref{eq:varestimate}) into
$\ell_{G,Y,i}$ and maximize the resulting profile log-likelihood
function
\begin{equation}
(\Omega_{i,\sib(i)},B_{i,\pa(i)}) \;\mapsto\; \log(N)-1 -\log\left( \frac{\|Y_i-B_{i,\pa(i)}Y_{\pa(i)}-\Omega_{i,\sib(i)}Z_{\sib(i)}\|^2}{(c_{i,0} + B_{i,\pa(i)}c_{i,\pa(i)})^2} \right).
\label{eq:max-profile}
\end{equation}
By monotonicity of the logarithm, maximizing~(\ref{eq:max-profile})
with respect to 
$(\Omega_{i,\sib(i)},B_{i,\pa(i)})\in
\mathbb{R}^{\sib(i)}\times\mathbb{R}^{\pa(i)}_{\text{inv}}$ is equivalent to
minimizing
\begin{equation} 
g_i(\Omega_{i,\sib(i)},B_{i,\pa(i)}) = \frac{\|Y_i-B_{i,\pa(i)}Y_{\pa(i)}-\Omega_{i,\sib(i)}Z_{\sib(i)}\|^2}{(c_{i,0} + B_{i,\pa(i)}c_{i,\pa(i)})^2}. \label{eq:min}
\end{equation}

If $c_{i,\pa(i)} = 0$, which occurs when $i$ does not lie on any
directed cycle, then the denominator in (\ref{eq:min}) is constant and
the problem amounts to finding least squares estimates for
$\Omega_{i,\sib(i)}$ and $B_{i,\pa(i)}$.  In other words, we solve a
linear regression problem with response $Y_i$ and covariates $Z_k$,
$k\in\sib(i)$ and $Y_j$, $j\in\pa(i)$.  This is the setting considered
by \cite{drton:eichler:richardson:2009}.  

In the more difficult case where $c_{i,\pa(i)} \neq 0$, minimizing the
function $g_i$ from~(\ref{eq:min}) amounts to minimizing a ratio of
two univariate quadratic functions.  The numerator is a least squares
objective for a linear regression problem with design matrix
$(Z_{\sib(i)}^T,Y_{\pa(i)}^T)\in\mathbb{R}^{N\times
  (|\sib(i)|+|\pa(i)|)}$.  The denominator is the square of an affine
function whose slope vector satisfies the following property proven in
Appendix~\ref{sec:proof-prop:block-update-special}.

\begin{lemma}
  \label{lem:block-update-special}
  The vector $\begin{pmatrix} 0\\ c_{i,\pa(i)} \end{pmatrix}$ is orthogonal to the kernel of
$\begin{pmatrix} Z_{\sib(i)} \\ Y_{\pa(i)} 
       \end{pmatrix}^T$.
\end{lemma}

\subsection{Minimizing a ratio of quadratic functions}
\label{sec:solving-block-update}

When $c_{i,\pa(i)} \neq 0$, the minimization of $g_i$
from~(\ref{eq:min}) is an instance of the general problem
\begin{align}
  \label{eq:block-update-general}
  \min_{\alpha\in\mathbb{R}^m} \frac{\|y-X\alpha\|^2}{(c_0+c^T \alpha)^2}
\end{align}
that is specified by a vector $y\in\mathbb{R}^N$ with $N\ge m$, a
matrix $X\in\mathbb{R}^{N\times m}$, a nonzero vector
$c\in\mathbb{R}^m\setminus\{0\}$ and a scalar $c_0\in\mathbb{R}$.  For
a correspondence to~(\ref{eq:min}), take as argument the vector
$\alpha =(\Omega_{i,\sib(i)},B_{i,\pa(i)})^T$, which is of length
$m=|\sib(i)|+|\pa(i)|$, and set
\begin{align}
  \label{eq:yXac}
    y &=Y_i^T,
  &
    X &=\begin{pmatrix} Z_{\sib(i)} \\ Y_{\pa(i)} 
       \end{pmatrix}^T, 
  &
    c &=\begin{pmatrix} 0\\ c_{i,\pa(i)} \end{pmatrix},
  &
    c_0&=c_{i,0}.
\end{align}
We now show that~(\ref{eq:block-update-general}) admits a closed-form
solution.  In doing so, we focus attention on problems in which the
matrix $X$ has full column rank.  Unless stated otherwise, we do not
require that $c$ be orthogonal to the kernel of $X$.  Rank deficient
cases are discussed in Remark~\ref{rem:not-unique} at the end of this
section.

\begin{theorem}
  \label{thm:main-BCD}
  Suppose the matrix $X$ has full rank $m\le N$.  Let
  $\hat\alpha=(X^TX)^{-1}X^Ty$ be the minimizer
  $\alpha\mapsto\|y-X\alpha\|^2$, and let
  $y_0^2=\|y-X\hat\alpha\|^2$. 
  \begin{enumerate}[label=(\roman{*}), ref=(\roman{*}),leftmargin=5.0em]
  \item If $c_0+c^T\hat\alpha\not=0$,
    then~(\ref{eq:block-update-general}) is uniquely solved by 
    \[
      \alpha^\star \;=\; \hat\alpha + \frac{y_0^2}{c_0+c^T\hat\alpha} (X^TX)^{-1}c.
    \]
  \item If $c_0+c^T\hat\alpha=0$ and $y_0^2=0$,
    then~(\ref{eq:block-update-general}) admits a solution, but not
    uniquely so.  The solution set is
    $\{\hat\alpha+\lambda (X^TX)^{-1}c:\lambda\in\mathbb{R}\setminus\{0\}\}$.
  \item If $c_0+c^T\hat\alpha=0$ and $y_0^2>0$, the minimum
    in~(\ref{eq:block-update-general}) is not achieved.  
  \end{enumerate}
\end{theorem}

\begin{remark}
  \label{rem:comp-effort}
  The computational complexity of
  solving~(\ref{eq:block-update-general}) is on the same order as that
  of solving the least squares problem with objective
  $\|y-X\alpha\|^2$.
\end{remark}

\begin{proof}[Proof of Theorem~\ref{thm:main-BCD}]
  We give a numerically stable algorithm for
  solving~(\ref{eq:block-update-general}), and then translate the
  solution into a rational function of the input $(y,X,c_0,c)$.

  \emph{(a) Algorithm.}  Find an orthogonal $m\times m$
  matrix $Q_1$ such that $Q_1 c=(0,\dots,0,\|c\|)^T$; note that in our
  context the support of $c$ is confined to the
  coordinates indexed by ${\pa(i)}$.  Reparametrizing to
  $\alpha'=Q_1\alpha$, (\ref{eq:block-update-general}) becomes
  \begin{align}
    \label{eq:block-update-general-Q}
    \min_{\alpha'\in\mathbb{R}^m}
    \frac{\|y-XQ_1^T\alpha'\|^2}{(c_0+\|c\|\,\alpha_m')^2}
  \end{align}
  with $\alpha_m'$ being the last coordinate of
  $\alpha'=(\alpha_1',\dots,\alpha_m')$.  Next, compute a QR
  decomposition $XQ_1^T= Q_2^TR$, where $Q_2$ is an orthogonal
  $N\times N$ matrix, and $R$ is an upper triangular $N\times m$
  matrix.  Observe that $R=\left(
    \begin{smallmatrix}
      R_1 \\ 0
    \end{smallmatrix}\right)$ with $R_1\in\mathbb{R}^{m\times m}$
  upper triangular.  Since orthogonal transformations
  leave Euclidean norms invariant,
  \begin{equation}
    \frac{\|y-XQ_1^T\alpha'\|^2}{(c_0+\|c\|\,\alpha_m')^2}
    \;=\;
    \frac{\|Q_2y-R \alpha'\|^2}{(c_0+\|c\|\,\alpha_m')^2}
    \;=\;
    \frac{\sum_{j=1}^m \left[(Q_2y)_j-(R_1 \alpha')_j\right]^2 +
      y_0^2}{(c_0+\|c\|\,\alpha_m')^2},
    \label{eq:minQtil} 
  \end{equation}
  where $y_0^2= \sum_{j=m+1}^N (Q_2y)_j^2$ is the squared length of the
  projection of $y$ on the orthogonal complement of the span of $X$.  
  Finally, we reparametrize to $\alpha''=R_1\alpha'$ and obtain the
  problem
  \begin{align}
    \label{eq:block-update-general-split}
    \min_{\alpha''\in\mathbb{R}^m}
    \frac{\sum_{j=1}^m \left[(Q_2y)_j-\alpha_j''\right]^2 +
    y_0^2}{(c_0+\|c\|r^{-1}\,\alpha_m'')^2}
  \end{align}
  with $r=R_{mm}$ being the $(m,m)$ entry in $R$ (and $R_1$).  We have
  $r\not=0$ as $X$ and thus $XQ_1^T$ and also $R$ have full column
  rank.  This also entails that $R_1$ is invertible.
  
  For
  $\alpha''$ to be a solution
  of~(\ref{eq:block-update-general-split}), it clearly must hold that
  \begin{equation}
    \alpha_j'' = (Q_2y)_j \quad\text{for}\quad j=1,\dots, m-1,
    \label{eq:soln1}
  \end{equation}
  and~(\ref{eq:block-update-general-split}) is solved by finding the
  coordinate $\alpha_m''$ by minimizing the univariate function
  \begin{equation} \label{eq:univariate_eqn}
    g(\alpha_m'') = \frac{\left((Q_2y)_m -
        \alpha_m''\right)^2+y_0^2}{(c_0 +
      \|c\|r^{-1}\,\alpha_m'')^2},  \qquad \alpha_m''\in\mathbb{R}.
  \end{equation}
  By Lemma \ref{lem:optimum} below and assuming that
  $c_0+\|c\|r^{-1}(Q_2y)_m\not=0$, the univariate function $g$ from
  \eqref{eq:univariate_eqn} attains its minimum at
  \begin{equation}
    \alpha_m'' = (Q_2y)_m + \frac{\|c\|y_0^2}{rc_0+\|c\|(Q_2y)_m}.
    \label{eq:soln2}
  \end{equation}
  If $c_0+\|c\|r^{-1}(Q_2y)_m=0$ and $y_0^2=0$, then $g$ is constant
  and any feasible $\alpha_m''\not= (Q_2y)_m$ is optimal. 
  If $c_0+\|c\|r^{-1}(Q_2y)_m=0$ and $y_0^2>0$, then $g$ does not
  achieve its minimum.

  In order to solve the problem posed at the beginning of this
  subsection, i.e., the problem from~(\ref{eq:block-update-general}), we
  convert the optimum $\alpha''$ from~(\ref{eq:soln1})
  and~(\ref{eq:soln2}) to 
  \begin{equation}
    \label{eq:alpha-opt-conversion}
    \alpha=Q_1^TR_1^{-1}\alpha''.
  \end{equation}

  \emph{(b) Rational formulas.}  Inspecting~(\ref{eq:minQtil}), we
  observe that $R_1^{-1}(Q_2y)_{\{1,\dots,m\}}$ is the coefficient
  vector that solves the least squares problem in which $y$ is
  regressed on $XQ_1^T$.  Therefore,
  \begin{equation}
    \label{eq:alpha-LS}
    Q_1^TR_1^{-1}(Q_2y)_{\{1,\dots,m\}}=(X^TX)^{-1}X^Ty=:\hat\alpha
  \end{equation}
  is the least squares coefficient vector for the regression of $y$ on
  $X$.  Because $R_1$ is rectangular, it follows that $r^{-1}(Q_2y)_m$
  is the $m$-th entry of the vector $R_1^{-1}(Q_2y)_{\{1,\dots,m\}}$.
  With $Q_1 c=(0,\dots,0,\|c\|)^T$, we deduce that
  \begin{equation}
    \label{eq:c+ad:rational}
    \|c\| r^{-1}(Q_2y)_m = \langle Q_1c,
    R_1^{-1}(Q_2y)_{\{1,\dots,m\}}\rangle = \langle c,
    Q_1^TR^{-1}(Q_2y)_{\{1,\dots,m\}}\rangle=\langle
    c,\hat\alpha\rangle.
  \end{equation}
  Let $e_m=(0,\dots,0,1)^T$ be the $m$-th canonical basis vector.
  Using that $R_1^{-T}$ has its last column equal to $r^{-1}e_m$, we
  find that
\begin{multline}
  \label{eq:offset}
  Q_1^TR_1^{-1}e_m \|c\|r^{-1}
  =Q_1^TR_1^{-1}R_1^{-T} Q_1c 
  =\left(Q_1^TR^TR Q_1\right)^{-1}c
  \\
  =\left(Q_1^TR^TQ_2Q_2^TR Q_1\right)^{-1}c
  =\left(X^TX\right)^{-1}c.
\end{multline}

In case (i), we obtain from~(\ref{eq:soln1}),~(\ref{eq:soln2})
and~(\ref{eq:alpha-opt-conversion}) that the unique minimum is
\[
  Q_1^TR_1^{-1}(Q_2y)_{\{1,\dots,m\}} + \frac{\|c\|r^{-1} y_0^2}{c_0+
    \|c\| r^{-1}(Q_2y)_m}Q_1^TR_1^{-1}e_m .
\]
Applying~(\ref{eq:alpha-LS})-(\ref{eq:offset}), we readily find the
rational formula asserted in the theorem.  Cases (ii) and (iii) are
similar. 
\end{proof}

The above proof relied on the following lemma about a ratio of
univariate quadratics.  The lemma is derived in
Appendix~\ref{app:ratio-lemma}.
  
  \begin{lemma} \label{lem:optimum} For constants $a,b,c_0,c_1\in\mathbb{R}$
    with $c_1\neq0$, define the function
    \[
      f(x) = \frac{(a-x)^2+b^2}{(c_0+c_1x)^2}, \qquad
      x\in\mathbb{R}\setminus\left\{-c_0/c_1\right\}. 
    \]
    \begin{enumerate}[label=(\roman{*}), ref=(\roman{*}),leftmargin=5.0em]
    \item If $c_0+ac_1 \neq 0$, then $f$ is uniquely minimized by
      \[
        x = \frac{ac_0+a^2c_1+b^2c_1}{c_0+ac_1} = a +
        \frac{b^2c_1}{c_0+ac_1}.
      \]
    \item   If $c_0+ac_1=0$ and
      $b=0$, then $f$ is constant and equal to $1/c_1^2$.
    \item 
      If $c_0+ac_1 = 0$ and $b^2>0$, then $f$ does not achieve its
      minimum, and
      $\inf f = \lim_{x\to\pm\infty} f(x) = 1/c_1^2$.
    \end{enumerate}
  \end{lemma}

\begin{remark}
  \label{rem:not-unique}
  When $c$ is orthogonal to the kernel of $X$, then $c=X^T\tilde c$
  for a vector $\tilde c\in\mathbb{R}^N$.  The
  problem~(\ref{eq:block-update-general}) is then equivalent to
  \begin{equation}
    \label{eq:block-update-implicit}
    \min_{\tilde\alpha\in\linspan(X)}
    \frac{\|y-\tilde\alpha\|^2}{(c_0+\tilde c^T\tilde\alpha)^2}.
  \end{equation}
  Let $\mathcal{L}(X)$ be the column span of $X$, and let
  $\pi_{\mathcal{L}(X)}$ be the orthogonal projection onto
  $\mathcal{L}(X)$.  Then~(\ref{eq:block-update-implicit}) admits a
  unique solution if and only if $c_0+\tilde c^T
  \pi_{\mathcal{L}(X)}(y)\not=0$.  The unique solution is
  \[
    \tilde\alpha^\star\;=\;
    \pi_{\mathcal{L}(X)}(y) + \frac{\|y-\pi_{\mathcal{L}(X)}(y)\|^2}{c_0+\tilde c^T
  \pi_{\mathcal{L}(X)}(y)}\, \pi_{\mathcal{L}(X)}(\tilde c),
  \]
  which is meaningful also when $X$ does not have full rank.  If
  desired, a coefficient vector $\alpha^\star\in\mathbb{R}^m$
  satisfying $X\alpha^\star=\tilde \alpha^\star$ can be chosen.  
\end{remark}

\subsection{The BCD algorithm}
\label{sec:runn-bcd-algor}

By Theorem~\ref{thm:main-BCD}, or rather the algorithm outlined in its
proof, we are able to efficiently minimize the function $g_i$
from~(\ref{eq:min}).  In other words, we can efficiently update the
$i$-th row in $B$ and the $i$-th row and column in $\Omega$ by a
partial maximization of the log-likelihood function $\ell_{G,Y}$.  We
summarize our block-coordinate descent scheme for maximization of the
log-likelihood function $\ell_{G,Y}$ in Algorithm~\ref{alg:BCD}.  For
a convergence criterion, we may compare the norm of the change in
$(B,\Omega)$ or the resulting covariance matrix or the value of
$\ell_{G,Y}$ to a given tolerance.

\begin{algorithm}[t]
\caption{Block-coordinate descent \label{alg:BCD}}
\begin{algorithmic}[1]
\Require $Y$, $\Omega^{(0)}$ and $B^{(0)}$
\Repeat
\For {$i \in V$}
	\State{Fix $\Omega_{-i,-i}$ and $B_{-i}$, and compute residuals $\epsilon_{-i}$ and pseudo-variables $Z_{\sib{(i)}}$} 
	\State {Compute $c_{i,0}$ and $c_{i,\pa(i)}$ as in Lemma~\ref{thm:detForm}}
	\If {$c_\pa(i) \neq 0$}
        \State{Set up problem~(\ref{eq:block-update-general}) with
          $y=Y_i^T$, $X=(Y_{\pa(i)}^T,Z_{\sib(i)}^T)$, 
          $c=(c_{i,\pa(i)}^T,0)^T$ and $c_0=c_{i,0}$}
        \State{Compute an orthogonal matrix $Q_1$ with
          $Q_1c=(0,\dots,0,\|c\|)^T$} 
        \State{Compute QR decomposition $Q_1X=Q_2^TR$}
        \State{Extract submatrix $R_1=R_{\{1,\dots,m\}\times
            \{1,\dots,m\}}$}
        \State{Compute intermediate constants  $r=R_{mm}$, $y_0^2$, and $(Q_2y)_j$ for $j=1,\ldots,m$}
        \State{Compute $\alpha''$ using (\ref{eq:soln1}) and (\ref{eq:soln2})}
        \State {Compute
          $(\hat B_{i,\pa(i)},\hat
          \Omega_{i,\sib(i)})^T=\alpha=Q_1^TR_1^{-1}\alpha''$}  
	\Else 
		\State{Compute $(\hat B_{i,\pa(i)},\hat
          \Omega_{i,\sib(i)})$ by minimizing sum of squares in numerator of \eqref{eq:min}}
	\EndIf
        \State{Compute $\hat \omega_{ii.-i}$
          using~(\ref{eq:varestimate}) }
        \State{Update $B_i$ and $\Omega_{i,-i}=\Omega_{-i,i}^T$ using
          $\hat B_{i,\pa(i)}$ and $\hat\Omega_{i,\sib(i)}$, respectively}
        \State {Update $\Omega$ by setting $\omega_{ii} = \hat \omega_{ii.-i}+
          \Omega_{i,-i}\Omega_{-i,-i}^{-1} \Omega_{-i,i}$}
\EndFor
\Until{Convergence criterion is met}
\end{algorithmic}
\end{algorithm}

Because cases (ii) and (iii) of Theorem~\ref{thm:main-BCD} allow for
non-unique or non-existent solutions to block update problems, a
remaining concern is whether the BCD algorithm may fail to be
well-defined.  We address this problem in
Section~\ref{sec:assumptionDiscussion}, where we give a
characterization of the mixed graphs for which BCD updates are unique
and feasible.
In this characterization we treat generic data $Y$ and generically
chosen starting values for $(B,\Omega)$.  As discussed in
Section~\ref{sec:block-ident}, graphs for which the BCD algorithm is
not generically well-defined yield non-identifiable models.
Identifiability is not necessary, however, for the BCD algorithm to be
generically well-defined.  Furthermore, note that non-uniqueness of
block update solutions could be addressed as outlined in
Remark~\ref{rem:not-unique}.

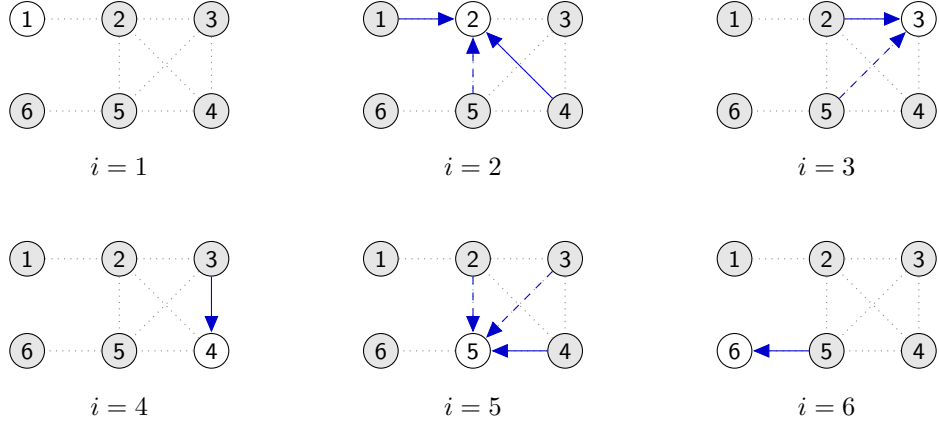
\begin{figure}[t]
  \centering
  \captionsetup[subfigure]{labelformat=empty,position=top}
  
\begin{subfigure}[t]{.3\linewidth}
        \centering

	\begin{tikzpicture}[->,>=triangle 45,shorten >=.2pt,
        auto,
        main node/.style={circle,inner
          sep=2pt,fill=gray!20,draw,font=\sffamily}]  
      
        \node[main node, fill=white] (1) {1}; 
        \node[main node] (2) [right=.75cm of 1] {2}; 
        \node[main node] (3) [right=.75cm of 2] {3}; 
        \node[main node] (4) [below=.75cm of 3] {4};
        \node[main node] (5) [left=.75cm of 4] {5};
      	\node[main node] (6) [left=.75cm of 5] {6};
      
        \path[color=gray,every
        node/.style={font=\sffamily\small}, style={-,dotted}] 
        (1) edge node {} (2)
        (2) edge node {} (3)
        (3) edge node {} (4)
        (4) edge node {} (5)
        (4) edge node {} (2)
        (5) edge node {} (6)
        (3) edge node {} (5)
        (2) edge node {} (5);
        \end{tikzpicture}
      \caption{$i=1$}
    \end{subfigure}
    ~
\begin{subfigure}[t]{.3\linewidth}
        \centering

	\begin{tikzpicture}[->,>=triangle 45,shorten >=.2pt,
        auto,
        main node/.style={circle,inner
          sep=2pt,fill=gray!20,draw,font=\sffamily}]  
      
        \node[main node] (1) {1}; 
        \node[main node, fill=white] (2) [right=.75cm of 1] {2}; 
        \node[main node] (3) [right=.75cm of 2] {3}; 
        \node[main node] (4) [below=.75cm of 3] {4};
        \node[main node] (5) [left=.75cm of 4] {5};
      	\node[main node] (6) [left=.75cm of 5] {6};
      
        \path[color=gray,style={-,dotted}] 
        (1) edge node {} (2)
        (2) edge node {} (3)
        (3) edge node {} (4)
        (4) edge node {} (5)
        (4) edge node {} (2)
        (5) edge node {} (6)
        (3) edge node {} (5)
        (2) edge node {} (5);
        
         \path[color=black!20!blue, style={->}] 
         (1) edge node {} (2)
         (4) edge node {} (2);
                
         \path[color=black!20!blue, style={->, dashed}]
         (5) edge node {} (2);
        \end{tikzpicture}
      \caption{$i=2$}
    \end{subfigure}
    ~
\begin{subfigure}[t]{.3\linewidth}
        \centering

\begin{tikzpicture}[->,>=triangle 45,shorten >=.2pt,
        auto,
        main node/.style={circle,inner
          sep=2pt,fill=gray!20,draw,font=\sffamily}]  
      
        \node[main node] (1) {1}; 
        \node[main node] (2) [right=.75cm of 1] {2}; 
        \node[main node, fill=white] (3) [right=.75cm of 2] {3}; 
        \node[main node] (4) [below=.75cm of 3] {4};
        \node[main node] (5) [left=.75cm of 4] {5};
      	\node[main node] (6) [left=.75cm of 5] {6};
      
        \path[color=gray,style={-,dotted}] 
        (1) edge node {} (2)
        (2) edge node {} (3)
        (3) edge node {} (4)
        (4) edge node {} (5)
        (4) edge node {} (2)
        (5) edge node {} (6)
        (3) edge node {} (5)
        (2) edge node {} (5);
        
         \path[color=black!20!blue, style={->}] 
         (2) edge node {} (3);
                
         \path[color=black!20!blue, style={->, dashed}]
         (5) edge node {} (3);
        \end{tikzpicture}
      \caption{$i=3$}
    \end{subfigure}
        \\ \vspace{.5cm}
          
  	\begin{subfigure}[t]{.3\linewidth}
                \centering
        
        \begin{tikzpicture}[->,>=triangle 45,shorten >=.2pt,
                auto,
                main node/.style={circle,inner
                  sep=2pt,fill=gray!20,draw,font=\sffamily}]  
              
                \node[main node] (1) {1}; 
                \node[main node] (2) [right=.75cm of 1] {2}; 
                \node[main node] (3) [right=.75cm of 2] {3}; 
                \node[main node, fill=white] (4) [below=.75cm of 3] {4};
                \node[main node] (5) [left=.75cm of 4] {5};
              	\node[main node] (6) [left=.75cm of 5] {6};
              
                \path[color=gray,style={-,dotted}] 
                (1) edge node {} (2)
                (2) edge node {} (3)
                (3) edge node {} (4)
                (4) edge node {} (5)
                (4) edge node {} (2)
                (5) edge node {} (6)
                (3) edge node {} (5)
                (2) edge node {} (5);
                
                 \path[color=black!20!blue, style={->}] 
                 (3) edge node {} (4);

                \end{tikzpicture}
              \caption{$i=4$}
            \end{subfigure}
            ~
        \begin{subfigure}[t]{.3\linewidth}
                \centering
				 \begin{tikzpicture}[->,>=triangle 45,shorten >=.2pt,  auto,
                       main node/.style={circle,inner
                         sep=2pt,fill=gray!20,draw,font=\sffamily}]  
                     
                       \node[main node] (1) {1}; 
                       \node[main node] (2) [right=.75cm of 1] {2}; 
                       \node[main node] (3) [right=.75cm of 2] {3}; 
                       \node[main node] (4) [below=.75cm of 3] {4};
                       \node[main node, fill=white] (5) [left=.75cm of 4] {5};
                     	\node[main node] (6) [left=.75cm of 5] {6};
                     
                       \path[color=gray,style={-,dotted}] 
                       (1) edge node {} (2)
                       (2) edge node {} (3)
                       (3) edge node {} (4)
                       (4) edge node {} (5)
                       (4) edge node {} (2)
                       (5) edge node {} (6)
                       (3) edge node {} (5)
                       (2) edge node {} (5);
                       
                        \path[color=black!20!blue, style={->}] 
                        (4) edge node {} (5)
                        ;
                         
                        \path[color=black!20!blue, style={->, dashed}]
                        (2) edge node {} (5)
                        (3) edge node {} (5)
                        ;      
         
                       \end{tikzpicture}
                       \caption{$i=5$}
        
            \end{subfigure}
            ~
            \begin{subfigure}[t]{.3\linewidth}
                    \centering
           \begin{tikzpicture}[->,>=triangle 45,shorten >=.2pt,  auto,
                                 main node/.style={circle,inner
                                   sep=2pt,fill=gray!20,draw,font=\sffamily}]  
                               
                                 \node[main node] (1) {1}; 
                                 \node[main node] (2) [right=.75cm of 1] {2}; 
                                 \node[main node] (3) [right=.75cm of 2] {3}; 
                                 \node[main node] (4) [below=.75cm of 3] {4};
                                 \node[main node] (5) [left=.75cm of 4] {5};
                               	\node[main node, fill=white] (6) [left=.75cm of 5] {6};
                               
                                 \path[color=gray,style={-,dotted}] 
                                 (1) edge node {} (2)
                                 (2) edge node {} (3)
                                 (3) edge node {} (4)
                                 (4) edge node {} (5)
                                 (4) edge node {} (2)
                                 (5) edge node {} (6)
                                 (3) edge node {} (5)
                                 (2) edge node {} (5);
                                 
                                  \path[color=black!20!blue, style={->}] 
                                  (5) edge node {} (6)
                                  ;
                   \end{tikzpicture}
                    \caption{$i=6$}
                \end{subfigure}

\caption{Illustration of the update steps for the BCD algorithm for each node. At each step, the edges corresponding to fixed parameters have been replaced with dotted edges. Arrowheads into nodes other than $i$ have been removed. Hence, solid and dashed directed edges into $i$ respectively represent directed and bi-directed edges with an arrowhead at $i$. Each remaining arrowhead signifies the relevant parameter to update during this step.}
  \label{fig:update_example}
\end{figure}

\begin{example}\rm
  \label{ex:run-bcd}
  We illustrate the BCD algorithm for the graph from Figure
  {\ref{fig:example}}, visiting the nodes in the order of their labels
  from 1 to 6.  Since the graph is simple (i.e., without bows), the
  theory from Section~\ref{sec:bcdAssump} shows that all updates are
  well-defined.

  Beginning with node $i=1$, we fix all but the first row of $B$ and
  the first row and column of $\Omega$.  In graphical terms, we fix
  the parameters that correspond to edges that do not have an
  arrowhead at node $1$.  Now, there are no arrowheads at node $1$,
  meaning that all entries in the first row of $B$ and all
  off-diagonal entries in the first row and column of $\Omega$ are
  constant zero.  Consequently, the algorithm merely updates the
  variance $\omega_{11}$.  The update simply sets
  $\omega_{11}=S_{11}$, the sample variance for variable 1.  This
  update is the same in later iterations, that is, node 1 can be
  skipped in subsequent iterations.

  For $i=2$, three edges have arrowheads at node 2, with corresponding
  parameters $\beta_{21}$, $\beta_{24}$ and $\omega_{25}$.  The
  directed edge $4 \to 2$ is contained in a cycle of the graph.  Its
  associated parameter, $\beta_{24}$, has coefficient
  $-\beta_{32}\beta_{43}$ in $\det(I-B)$. Thus, unless $\beta_{32}$ or
  $\beta_{43}$ is zero, $c_{\pa(2)} \neq 0$ and the more involved
  update from lines 6-11 in Algorithm~\ref{alg:BCD} applies.  If
  $\beta_{32}$ or $\beta_{43}$ is fixed to zero during this first
  iteration of the algorithm (i.e., one or both were initialized to
  zero), then the first update for $i=2$ is a least squares problem.

  Nodes $3$ and $4$ each have one arrowhead corresponding to a
  directed edge contained in a cycle of the graph. Hence, the updates
  for $i=3$ and $i=4$ proceed analogously to the update step
  $i=2$. For $i=3$, we update the parameters
  $\beta_{32}, \omega_{35}$, and $\omega_{33}$. For $i=4$, we update
  the parameters $\beta_{43}$ and $\omega_{44}$.

  For $i=5$, there are three arrowheads at node 5 corresponding to the
  three parameters $\beta_{54}, \omega_{25},$ and
  $\omega_{35}$. Observe that $4\to 5$ is the only directed edge into
  node 5 and is not contained in a cycle. Hence $c_{\pa(5)} = 0$, and
  we proceed with the least squares update in line 13 of
  Algorithm~\ref{alg:BCD}.  This least squares computation may change
  from one iteration of the algorithm to the next.

  For $i=6$, the only arrowhead corresponds to the directed edge
  $5\to 6$ with associated parameter $\beta_{65}$.  This directed edge
  is not involved in a cycle, so we estimate the parameter via a least
  squares regression and then solve for $\omega_{66}$.  This update
  remains the same throughout all iterations of the algorithm and only
  needs to be performed once.
\end{example}

\section{Properties of the block-coordinate descent algorithm}
\label{sec:assumptionDiscussion}

\subsection{Convergence properties}

Because the BCD algorithm performs partial maximizations, the value of
the log-likelihood function $\ell_{G,Y}$ is non-decreasing throughout the
iterations.  At every update, the algorithm finds a positive
definite covariance matrix.  The update steps clearly preserve the
structural zeros of the matrices $B$ and $\Omega$, and
$I-B$ remains invertible.  Hence, the algorithm constructs a sequence
in $\mathbf{B}(G)\times \mathbf{\Omega}(G)$.

Every accumulation point $(B^\star,\Omega^\star)$ of the sequence
constructed by the algorithm is a critical point of the likelihood
function and either a local maximum or a saddle point.  A local
maximum can be certified by checking negative definiteness of the
Hessian of $\ell_{G,Y}$.  However, as `always' in general non-linear
optimization there is no guarantee that a global maximum is found.
Indeed, even for seemingly simple mixed graphs, the likelihood
function can be multimodal \citep{drton:richardson:2004}.  In
practice, one may wish to run the algorithm from several different
initial values.
A strength of the BCD algorithm is that for nodes whose incoming
directed edges are not contained in any cycle of $G$ and that are not
incident to any bi-directed edges, the update of $B_{i,\pa(i)}$ and
$\omega_{ii}$ does not depend on the fixed pair
$(B_{-i},\Omega_{-i,-i})$ and thus needs to performed only once (in
the first iteration).  As we had noted, this happens for nodes 1 and 6
of the example discussed in Section~\ref{sec:runn-bcd-algor}.  Hence,
we may check for nodes of this type and exclude them from subsequent
iterations after the first iteration of the algorithm.  We also update
these nodes before the set of nodes that require multiple update
iterations.

\subsection{Existence and uniqueness of optima in block updates}
\label{sec:bcdAssump}

The BCD algorithm is well-defined if each block update problem has a
unique solution that is feasible, where feasibility refers to the new
matrix $\Omega$ being positive definite.  When updating at node $i$,
the positive definiteness of $\Omega$ is equivalent to
$\omega_{ii.-i}>0$.  Since the latter conditional variance is set
via~(\ref{eq:varestimate}), feasibility of a block update solution
$(\Omega_{i,\sib(i)},B_{i,\pa(i)})$ corresponds to
$\|Y_i-B_{i,\pa(i)}Y_{\pa(i)}-\Omega_{i,\sib(i)}Z_{\sib(i)}\|$ being
positive.

If the underlying graph is acyclic then the update at node $i$ solves
a least squares problem that has a unique solution if and only if the
$|\pa(i)|+|\sib(i)|$ vectors in the rows of $Y_{\pa(i)}$ and
$Z_{\sib(i)}$ form a linearly independent set in $\mathbb{R}^N$.
Moreover, the update yields a positive value of $\omega_{ii.-i}$ if
and only if $Y_i$ is not in the linear span of the rows of
$Y_{\pa(i)}$ and $Z_{\sib(i)}$.  We conclude that, in the acyclic
case, the block update admits a unique and feasible solution if and
only if the following condition is met:
\begin{enumerate}[label=(A\arabic{*})$_i$,
  ref=(A\arabic{*})$_i$,leftmargin=5.0em] 
\item \label{A1} The matrix $\begin{pmatrix} 
    Z_{\sib(i)}\\ Y_{\pa(i)\cup\{i\}} 
     \end{pmatrix}\in\mathbb{R}^{(|\sib(i)|+|\pa(i)|+1)\times
    N}$ has linearly independent rows.
\end{enumerate}
As we show in Theorem~\ref{thm:A1A2} below, if the underlying graph is
not acyclic, then a further condition is needed:
\begin{enumerate}[label=(A\arabic{*})$_i$,
  ref=(A\arabic{*})$_i$,leftmargin=5.0em]
  \setcounter{enumi}{1}
\item \label{A2} The inequality
  $c_{i,0}+\hat B_{i,\pa(i)} c_{i,\pa(i)}\not=0$ holds for
  $\hat B_{i,\pa(i)}=\left[Y_iX_i^T(X_iX_i^T)^{-1}\right]_{\pa(i)}$ and
  $X_i=\begin{pmatrix} Z_{\sib(i)}\\ Y_{\pa(i)}
     \end{pmatrix}\in\mathbb{R}^{(|\sib(i)|+|\pa(i)|)\times
    N}$.   
\end{enumerate}
Note that the acyclic case has $c_{i,0}=1$ and $c_{i,\pa(i)}=0$, so
condition~\ref{A2} is void.

\begin{example}\rm
  \label{ex:2cycle-nonexist}
  Let the graph $G=(V,E_{\to},E_{\bi})$ be a two-cycle, so
  $V=\{1,2\}$, $E_{\to}=\{1\to 2,\, 2\to 1\}$ and $E_{\bi}=\emptyset$.
  Consider the update for node $i=2$.  With $\pa(2)=1$, we have
  $c_{2,\pa(2)}=-\beta_{12}$ and $c_{2,0}=1$.  Since $\sib(2)=\emptyset$,
  the block update amounts to solving
  \[
    \min_{\beta_{21}\in\mathbb{R}}
    \frac{\|Y_2-\beta_{21}Y_1\|^2}{(1-\beta_{12}\beta_{21})^2} 
  \]
  for fixed $\beta_{12}$.  Condition~\ref{A1} holds for $i=2$ when the
  data vectors $Y_1$ and and $Y_2$ are linearly independent.  We are
  then in case (i) or (iii) of Theorem~\ref{thm:main-BCD}.  Hence, the
  solution either exists uniquely or does not exist.  It fails to
  exist when
  \[
    1-\beta_{12}\frac{\langle Y_2,Y_1\rangle}{\|Y_1\|^2} \;=\;0,
  \]
  that is, when~\ref{A2} fails for $i=2$.
\end{example}

\begin{theorem}
  \label{thm:A1A2}
  Let $G=(V,E_{\to},E_{\bi})$ be any mixed graph, and let
  $Y\in\mathbb{R}^{V\times N}$ be a data matrix of full rank
  $|V|\le N$.  Let $i\in V$ be any node.  Then the function $g_i$
  from~(\ref{eq:min}) has a unique minimizer
  $(\Omega_{i,\sib(i)},B_{i,\pa(i)})$ with
  \[
    \|Y_i-B_{i,\pa(i)}Y_{\pa(i)}-\Omega_{i,\sib(i)}Z_{\sib(i)}\|\;>\;0
  \]
  if and only if conditions~\ref{A1} and~\ref{A2} hold.
\end{theorem}
\begin{proof}
  \emph{$(\Longleftarrow)$} When~\ref{A1} holds,
  Theorem~\ref{thm:main-BCD} applies to the minimization of $g_i$
  because the matrix $X$ defined in~(\ref{eq:yXac}) has full rank.
  Condition~\ref{A2} ensures we are in case (i) of the theorem.
  Hence, $g_i$ has a unique minimizer
  $(\Omega_{i,\sib(i)},B_{i,\pa(i)})$.  According to~\ref{A1},
  $Y_i^T$ is not in the span of $X$.  Thus,
  $Y_i-B_{i,\pa(i)}Y_{\pa(i)}-\Omega_{i,\sib(i)}Z_{\sib(i)}\not=0$.
  \smallskip

  \emph{$(\Longrightarrow)$} First, suppose~\ref{A1} holds
  but~\ref{A2} fails.  Then Theorem~\ref{thm:main-BCD} applies in
  either case (ii) or (iii).  Hence, the minimizer of $g_i$ is either
  not unique or does not exist.

  Second, suppose condition~\ref{A1} fails because
  $X=(Z_{\sib(i)}^T, Y_{\pa(i)}^T)$ is not of full rank.  Let
  $\eta\in\mathbb{R}^{|\sib(i)|+|\pa(i)|}$ be any nonzero vector in
  the kernel of $X$.  Let $c=(0,c_{i,\pa(i)}^T)^T$.  With the
  orthogonality from Lemma~\ref{lem:block-update-special}, we have
  $X\alpha=X(\alpha+\eta)$ and $c^T\alpha=c^T(\alpha+\eta)$ for any
  $\alpha\in\mathbb{R}^{|\sib(i)|+|\pa(i)|}$.  Consequently, $g_i$
  does not have a unique minimizer.
  
  Third, suppose that $X=(Z_{\sib(i)}^T, Y_{\pa(i)}^T)$ has full rank
  but~\ref{A1} still fails.  Then $y=Y_i^T$ is in the column span of
  $X$ so that Theorem~\ref{thm:main-BCD} applies with the quantity
  $y_0^2$ zero.  We are thus in either case (i) or case (ii) of the
  theorem.  In case (ii) the minimizer is not unique.  This leaves us
  with case (i), in which $y_0^2=0$ implies that $g_i$ is uniquely
  minimized by the least squares vector $\hat\alpha$, i.e., the
  minimizer of $\alpha\mapsto\|y-X\alpha\|^2$.  Since $y=Y_i^T$ is in
  the span of $X$, we have $\|y-X\alpha\|^2=0$, which translates into
  $Y_i-B_{i,\pa(i)}Y_{\pa(i)}-\Omega_{i,\sib(i)}Z_{\sib(i)}=0$.  We
  conclude that $g_i$ has a unique and feasible minimizer only
  if~\ref{A1} and~\ref{A2} hold.
\end{proof}

\begin{example}
  \rm
  \label{ex:twonodes:threeedges}
  Let $G=(V,E_{\to},E_{\bi})$ be the graph with vertex set
  $V=\{1,2\}$, and edge sets $E_{\to}=\{1\to 2,\, 2\to 1\}$ and
  $E_{\bi}=\{1\bi 2\}$.  Note that the model $\mathbf{N}(G)$ comprises
  all centered bivariate normal distributions.  Therefore, the
  log-likelihood function $\ell_{G,Y}$ achieves its maximum for any data
  matrix $Y\in\mathbb{R}^{2\times N}$ of rank $2$.

  The two block updates in this example are symmetric, so consider the
  update for $i=1$ only.  Fix any two values of
  $\beta_{21}\in\mathbb{R}$ and $\omega_{22}>0$.  Then the map from
  $(\beta_{12},\omega_{12},\omega_{11})$ to the covariance matrix
  $(I-B)^{-1}\Omega(I-B)^{-T}$ is easily seen to have a Jacobian
  matrix of rank 2.  Because the rank drops from 3 to 2, for each
  triple $(\beta_{12},\omega_{12},\omega_{11})$ there is a
  one-dimensional set of other triples that yield the same covariance
  matrix and, thus, the same value of the likelihood function.  Due to
  this lack of block-wise identifiability, the block update cannot
  have a unique solution.

  In this example, we have $\sib(1)=\pa(1)=\{2\}$ and
  $\det(I-B)=1-\beta_{12}\beta_{21}$, so that $c_{i,0}=1$ and
  $c=(0,-\beta_{21})^T$.  Moreover,
  \[
    X^T=\begin{pmatrix}
     Z_{\sib(i)}\\ Y_{\pa(i)}
    \end{pmatrix}
     \;=\;
    \begin{pmatrix}
      \frac{1}{\omega_{22}} (Y_2-\beta_{21}Y_1)\\ Y_2
    \end{pmatrix}
    \;=\;
    \begin{pmatrix}
      -\frac{\beta_{21}}{\omega_{22}} & \frac{1}{\omega_{22}}\\
      0 & 1
    \end{pmatrix}
    Y.
  \]
  If $\beta_{21}=0$, then~\ref{A1} fails for $i=1$ because $X$ is rank
  deficient.  If $\beta_{21}\not=0$ and $\rank(Y)=2$, then
  $\rank(X)=2$ and $y=Y_1^T$ is in the span of $X$, with
  \[
    Y_1 =
    \begin{pmatrix}
       -\frac{\omega_{22}}{\beta_{21}} & \frac{1}{\beta_{21}}
    \end{pmatrix}
    X^T
    .
  \]
  Consequently, $y_0^2=0$ and the least squares coefficients for the
  regression of $y$ on $X$ are
  $(-\omega_{22}/\beta_{21},1/\beta_{21})$.  Then condition~\ref{A2}
  fails for $i=1$ because with $c_{1,\pa(1)}=-\beta_{21}$ and least
  squares coefficient $\hat B_{1,\pa(1)}=1/\beta_{21}$ we find that
  \[
    c_{1,0}+\hat B_{1,\pa(1)} c_{1,\pa(1)} = 1 +
    \frac{1}{\beta_{21}}(-\beta_{21}) = 0.
  \]
\end{example}

\begin{remark}
  \label{rem:y02zeroA2fails}
  The findings from Example~\ref{ex:twonodes:threeedges} generalize.
  Indeed, for any graph $G$, if $Y$ has full rank and $Y_i^T$ is in
  the span of $X=(Z_{\sib(i)}^T, Y_{\pa(i)}^T)$, then one can show
  that~\ref{A2} fails and, thus, the block update has
  infinitely many solutions; see
  Appendix~\ref{app:proof:rem:y02zeroA2fails}.
\end{remark}

\subsection{Well-defined BCD iterations}
\label{sec:well-defined-bcd}

Although Theorem~\ref{thm:A1A2} characterizes the existence of a
unique and feasible solution for a particular block update, it does
not yet clarify when its conditions~\ref{A1} and~\ref{A2} hold
throughout all iterations of the BCD algorithm.  In practice, there is
freedom in choosing the starting value
$(B_0,\Omega_0)\in\mathbf{B}(G)\times \mathbf{\Omega}(G)$ and, in
particular, we may choose it randomly to alleviate problems of having
the triple $(Y,B_0,\Omega_0)$ in undesired special position; recall
Example~\ref{ex:2cycle-nonexist}.  Since our models consider a
continuously distributed data matrix $Y\in\mathbb{R}^{V\times N}$, the
natural problem becomes to characterize the graphs $G$ such that any
finite number of BCD iterations are well-defined for generic triples
$(Y,B_0,\Omega_0)$.  As before, our treatment assumes
$N\ge |V|$.

We begin by studying condition~\ref{A1}.  
Let $G=(V,E_{\to},E_{\bi})$ be a mixed graph.  Let $\pi$ be a path in
$G$, and let $i_1,\dots,i_k$ be the not necessarily distinct vertices
on $\pi$.  Then $\pi$ is a \emph{half-collider path} if either all
edges on $\pi$ are bi-directed, or the first edges is $i_1\to i_2$ and
all other edges are bi-directed.  Both a single edge $i_1\to i_2$ and
an empty path comprising only node $i_1$ are half-collider paths.  The
\emph{bi-directed portion} of a half-collider path $\pi$ is the set of
nodes that are incident to a bi-directed edge on $\pi$.  In other
words, if $\pi$ starts with $i_1\to i_2$, then its bi-directed portion
is $\{i_2,\dots,i_k\}$.  If $\pi$ does not contain a directed edge,
then its bi-directed portion is the set of all of its nodes
$\{i_1,\dots,i_k\}$.  Valid half-collider paths are shown in Figure
\ref{eq:coHalfTrek}. 

We note that half-collider paths are dual to the half-treks of
\citet{foygel:draisma:drton:2012}.  A half-trek is a path whose
first edge is either directed or bi-directed, and whose remaining edges
are directed.

\begin{figure}[t]
\normalsize
\begin{align*}
i_1 \rightarrow \fbox{$i_2 \leftrightarrow i_3 \leftrightarrow \ldots
  \leftrightarrow i_k$}\,,  && i_1 \rightarrow \fbox{$i_2$}\,, &&
\fbox{$i_1 \leftrightarrow i_2 \leftrightarrow i_3 \ldots
                                                              \leftrightarrow i_k$}\,, &&  \fbox{$i_1$}\,.
\end{align*}
\caption{\label{eq:coHalfTrek} Four half-collider paths with boxes drawn
  around their bi-directed portions.}
\end{figure}

Let $S_b,S_e\subset V$ be two sets of nodes.  A collection of paths
$\pi^1,\dots,\pi^s$ is a \emph{system of half-collider paths} from $S_b$ to
$S_e$ if $|S_b|=|S_e|=s$, each $\pi^l$ is a half-collider path from a
node in $S_b$ to a node in $S_e$, every node in $S_b$ is the first
node on some $\pi^l$, and every node in $S_e$ is the last node on some
$\pi^l$.  

\begin{proposition} \label{lem:fullRankX} Let $G=(V,E_{\to},E_{\bi})$
  be a mixed graph, and let $i\in V$.  Then the following two
  statements are equivalent:
  \begin{enumerate}
  \item[(a)] Condition~\ref{A1} holds for generic triples
    $(Y,B,\Omega)\in\mathbb{R}^{V\times
      N}\times\mathbf{B}(G)\times\mathbf{\Omega}(G)$.
  \item[(b)] The induced subgraph $G_{-i}$ contains a system of
    half-collider paths from a subset of
    $V \setminus (\pa(i)\cup\{i\})$ to $\sib(i)$ such that the
    bi-directed portions are pairwise disjoint.
  \end{enumerate}
\end{proposition}

The proof is deferred to Appendix
\ref{sec:fullRankXProof}.  It merely requires $Y$ to be of full
rank and $(B,\Omega)$ to be chosen from a set of generic points that
is independent of $Y$.

\begin{example}\rm
  Suppose a graph with vertex set $V=\{1,\dots,6\}$ contains the
  paths
  \[
    1\to 3 \bi 4 \bi 5 \quad\text{and}\quad   2\bi 1 \bi 6.
  \]
  These form a system of half-collider paths from $\{1,2\}$ to
  $\{5,6\}$.  The system is not vertex disjoint as node 1 appears on
  both paths.  However, the bi-directed portions $\{3,4,5\}$ and
  $\{1,2,6\}$ are disjoint.
\end{example}

Next, we turn to condition~\ref{A2} and show that in generic cases it
does not impose any additional restriction.

\begin{proposition}
  \label{prop:A2generic}
  Suppose the mixed graph $G$ is such that~\ref{A1} holds
  for generic triples
  $(Y,B,\Omega)\in\mathbb{R}^{V\times
    N}\times\mathbf{B}(G)\times\mathbf{\Omega}(G)$.  Then~\ref{A2}
  holds for generic triples $(Y,B,\Omega)$.
\end{proposition}
\begin{proof}
  The matrix $X_i$ and the least squares vector $\hat B_{i,\pa(i)}$ in
  condition~\ref{A2} are rational functions of the triple
  $(Y,B,\Omega)$.  Hence, there is a polynomial $f(Y,B,\Omega)$ such
  that~\ref{A2} fails only if $f$ vanishes.  A polynomial that is not
  the zero polynomial has a zero set that is of reduced dimension and
  of measure zero \citep[Lemma 1]{okamoto:1973}.  Therefore, it
  suffices to show that~\ref{A2} holds for a single choice of
  $(Y,B,\Omega)$.
  
  By assumption, we may pick $B\in\mathbf{B}(G)$ and
  $\Omega\in\mathbf{\Omega}(G)$ such that~\ref{A1} holds for any full
  rank $Y$.  Take $Y$ such that
  $(I-B)^{-1}\Omega(I-B)^{-T} = \frac{1}{N}YY^T$.
  When $Y$ has full rank, the normal distribution with covariance
  matrix $\frac{1}{N}YY^T$ has maximal likelihood.  Therefore,
  $\Omega$ and $B$ are maximizers of the log-likelihood function
  $\ell_{G,Y}$.  Consider now the block update for node $i$.
  Because~\ref{A1} holds, the matrix $X_i$ has full rank and $Y_i$ is
  not in the span of $X_i$.  It follows that
  Theorem~\ref{thm:main-BCD} applies with $y_0^2>0$.  Since our
  special choice of $(Y,B,\Omega)$ guarantees the existence of an
  optimal solution, we must be in case (i) of the theorem.  The
  inequality defining this case corresponds to~\ref{A2}.
\end{proof}

The following theorem gives a combinatorial characterization of the
graphs for which the BCD algorithm is well-defined.  It readily
follows from the above results, as we show in
Appendix \ref{sec:unique_updatesProof}.

\begin{theorem} \label{thm:unique_updates} For a mixed graph
  $G=(V,E_{\to},E_{\bi})$, the following two statements are equivalent:
  \begin{enumerate}
  \item[(a)] For all $i\in V$, the induced subgraph $G_{-i}$ contains
    a system of half-collider paths from a subset of
    $V \setminus (\pa(i)\cup\{i\})$ to $\sib(i)$ such that the
    bi-directed portions are pairwise disjoint.
  \item[(b)] For generic triples
    $(Y,B_0,\Omega_0)\in\mathbb{R}^{V\times
      N}\times\mathbf{B}(G)\times\mathbf{\Omega}(G)$, any finite
    number of iterations of the BCD algorithm for $\ell_{G,Y}$ have
    unique and feasible block updates when 
    $(B_0,\Omega_0)$ is used as starting value.
  \end{enumerate}
\end{theorem}

In \cite{drton:eichler:richardson:2009}, the focus was on bow-free
acyclic graphs, where bow-free means that there do not exist two nodes
$i$ and $j$ with both $i\to j$ and $i\bi j$ in $G$.  For such graphs,
the BCD algorithm is easily seen to be well-defined.  More generally,
by taking $S_i=\sib(i)$ we obtain the following generalization to
graphs that may contain directed cycles.

\begin{proposition}
  \label{prop:bow-free}
  If $G$ is a simple mixed graph, i.e., every pair of nodes is
  incident to at most one edge, then condition (a) in
  Theorem~\ref{thm:unique_updates} holds.
\end{proposition}

When the graph $G$ is not simple, checking condition (a) from
Theorem~\ref{thm:unique_updates} is more involved.  It can, however,
be checked in polynomial time.

\begin{proposition}
  \label{prop:poly-time}
  For any mixed graph $G=(V,E_{\to},E_{\bi})$, condition (a) in
  Theorem~\ref{thm:unique_updates} can be checked in
  $\mathcal{O}(|V|^5)$ operations.
\end{proposition}

The proof, which is deferred to Appendix~\ref{sec:poly-time}, casts
checking the condition as a network flow problem.

\subsection{Identifiability}
\label{sec:block-ident}

There is a close connection between well-defined block updates and
parameter identifiability.  Suppose the data matrix $Y$ is such that
the sample covariance is
$S=\frac{1}{N}YY^T = (I-B)^{-1}\Omega(I-B)^{-T}$ for a pair
$(B,\Omega)\in\mathbf{B}(G)\times \mathbf{\Omega}(G)$.  Consider the
block update of the $i$-th row of $B$ and $i$-th row and column of
$\Omega$.  Based on Theorem~\ref{thm:main-BCD}, if the update does
not have a unique solution then there is an infinite set of solutions
$(B',\Omega')$.  Each such solution $(B',\Omega')$ must have
$(I-B')^{-1}\Omega'(I-B')^{-T}$ equal to $S$ because $S$ is the unique
covariance matrix with maximum likelihood.  Hence, there is an
infinite set of parameters $(B',\Omega')$ that define the same
normal distribution as $(B,\Omega)$.  

\begin{corollary}
 \label{cor:identifiable}  
 If the graphical condition in statement (a) of
 Theorem~\ref{thm:unique_updates} fails for the graph $G$, then the
 from Lemma \ref{lem:fullRankX} does not hold for any $i \in V$, then
 the parameters of model $\mathbf{N}(G)$ are not identifiable.
\end{corollary}


\section{Simulation studies}
\label{sec:simulations}

In this section, we analyze the performance of our BCD algorithm in
two contexts. First, we use it to compare the fit of two nested models
(one of which is cyclic) for data on protein abundances. Second, we
examine the problem of parameter estimation in a specified model.
There we compare our algorithm on a number of simulated graphs against
the fitting routine from the `sem' package in R \citep{fox:2006,
  R:2011}.

\subsection{Protein-signaling network}\label{sec:sachs_data}
Figure 2 in \cite{sachs:2005} presents a protein-signaling network
involving 24 molecules.  Abundance measurements are available for 11
of these.  The remaining 13 are unobserved.  For our illustration, we
select two plausible mixed graphs over the 11 observed variables.  The
graphs differ only by the presence of a directed edge that induces a
cycle and a bow; see Figure \ref{fig:protein_network}.  The
edge PIP2 $\to$ PIP3, which makes for the difference, is
highlighted in red.  Before proceeding to our analysis, we note that
the results in \citet{sachs:2005} are based on discretized data and
are thus not directly comparable to our computations.

\begin{figure}[t]
  \centering

	\begin{tikzpicture}[->,>=triangle 45,shorten >=1pt,
            auto,
            main node/.style={ellipse,inner sep=0pt,fill=gray!20,draw,font=\sffamily,
             minimum width = 1.2cm, minimum height = .7cm}]  
          
            \node[main node] (PKC) {PKC};
            \node[main node] (JNK) [below=3cm of PKC] {JNK};
			\node[main node] (PIP2) [left=3cm of JNK] {PIP2};
            \node[main node] (PLCg) [above=1cm of PIP2] {PLCg}; 
          	\node[main node] (PIP3) [right=2cm of PLCg] {PIP3}; 
          	\node[main node] (Akt) [right=1.5cm of PIP3] {Akt};
          	\node[main node] (P38) [right=1.5cm of Akt] {P38};
          	\node[main node] (Raf) [right=1.8cm of P38] {Raf};
          	\node[main node] (PKA) at (P38 |- JNK) {PKA};
          	\node[main node] (Mek) at (PKA -| Raf) {Mek};
          	\node[main node] (Erk) [below=.75cm of Mek] {Erk};

          	 \path[color=black,style={->}]
          	 (PKC) edge node {} (JNK)
          	 (PKC) edge node {} (Raf)
          	 (PKC) edge node {} (P38)
          	 ;
          	 
          	 \path[color=black,style={->}]
          	 (PLCg) edge node {} (PKC)
          	 (PLCg) edge node {} (PIP2);  
				
			\path[color=black,style={->}]
			(PIP3) edge node {} (PLCg)
			(PIP3) edge node {} (Akt)
			;
       
       		\path[color=black,style={->}]
       		(Raf) edge node {} (Mek)
       		;
       		
       		\path[color=black,style={->}]
       		(PIP2) edge node {} (PKC)
       		;
       		
       		\path[color=black!20!red,style={->, bend right}]
       		(PIP2) edge node {} (PIP3)
       		;
       
       		\path[color=black,style={->}]
       		(PKA) edge node {} (JNK)
       		(PKA) edge node {} (Akt)
       		(PKA) edge node {} (P38)
       		(PKA) edge node {} (Raf)
       		(PKA) edge node {} (Erk)
       		;
       
       		\path[color=black,style={->}]
       		(Mek) edge node {} (Erk)
       		;
       		
       		\path[color=black,style={<->, dashed}]
			(PIP2) edge node {} (PIP3)
			(Raf) edge node {} (PIP2)
       		;
       		
       		\path[color=black,style={<->, dashed, bend right}]
       					(Raf) edge node {} (PIP3)
       		       		;
       
            \end{tikzpicture}
  \caption{Plausible mixed graph for the protein-signaling network dataset. The relevant acyclic sub-model can be formed by removing the red directed edge from PIP2 to PIP3.}
  \label{fig:protein_network}
\end{figure}
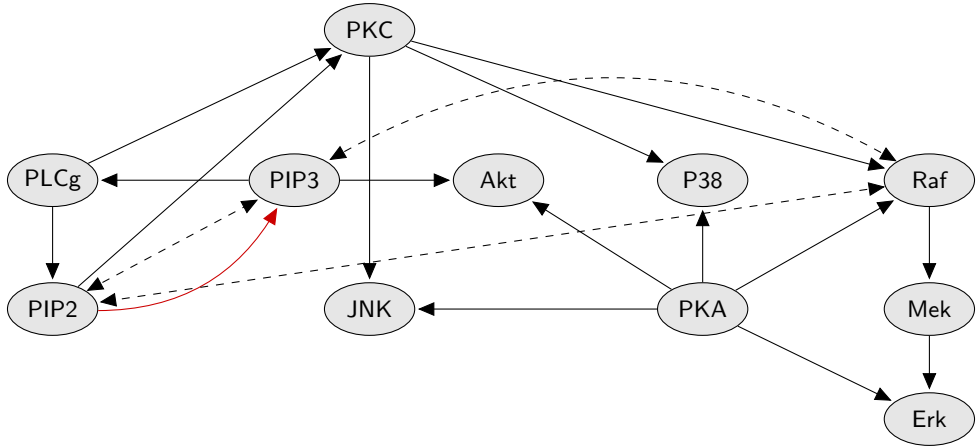

We proceed by comparing the two candidate models via the likelihood
ratio test. The data we consider consist of 11 simultaneously observed
signaling molecules measured independently across $N=853$ individual
primary human immune system cells. Specifically, we consider the data
from experimental condition CD3+CD28 and center/rescale the data,
ensuring that each variable has zero mean and variance one. Although
the likelihood ratio test statistic is invariant to scale, the
rescaling improves the conditioning of the sample covariance matrix
which improves the performance of BCD.

The corresponding likelihood ratio test statistic for the data is
.075, and under the standard $\chi^2_1$ asymptotic distribution for
the null hypothesis, this corresponds to a p-value of 0.78.  However,
in the considered models it is not immediately clear whether a
$\chi^2_1$ approximation has (asymptotic) validity, as the models
generally have a singular parameter space \citep{drton:2009}.
Therefore, we enlist subsampling as a guard against a possible
non-standard asymptotic distribution. Subsampling only requires the
existence of a limiting distribution for the likelihood ratio
statistic \citep[Chapter 2.6]{subsampling}. This limiting
distribution, while not necessarily chi-squared, is guaranteed to
always exist \citep{drton:2009}. Each random subsample consists of $b$
observations where $b$ is chosen large enough to approximate the true
asymptotic distribution under the null, but small compared to $N=853$
to still provide reasonable power under the alternative. We consider
5000 subsamples of sizes $b = 30$ and $b=50$.

\begin{figure}[t]
  \centering
  \includegraphics[scale=.5]{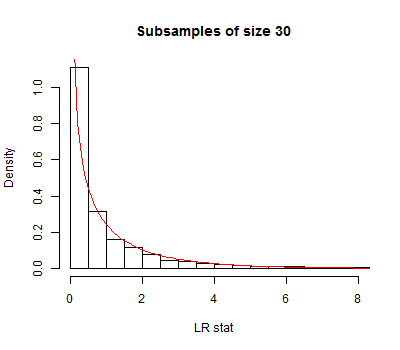}
  \includegraphics[scale=.5]{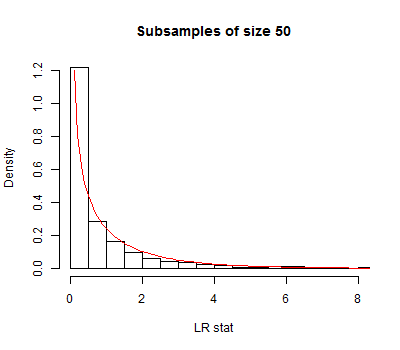}
  \caption{Histograms for the likelihood ratio test statistic for 5000
    subsamples of size 30 and 50, respectively. The superimposed red
    line depicts the $\chi^2_1$ density.}
  \label{fig:histograms}
\end{figure}

For each subsample, we first fit the sub-model corresponding to the
mixed graph depicted in Figure \ref{fig:protein_network} without the
edge PIP2 $\rightarrow$ PIP3.  For this procedure, we initialize the
free entries of $B$ using least squares regression estimates (i.e.,
fitting the model that ignores
the error correlations). We then calculate
the covariance between the regression residuals to estimate the
non-zero elements of $\Omega$. Although the sample covariance of the
regression residuals is positive definite, the resulting matrix which
also encodes the structural zeros may not be. To ensure that $\Omega$
is positive definite, we scale the off diagonal elements such that
$\sum_{i\neq j} |\omega_{ij}| = .9\times \omega_{ii}$ so that the
resulting matrix is diagonally dominant. After the BCD algorithm
converges to a stationary point in the sub-model, we take the fitted
values $\hat{B}$ and $\hat{\Omega}$ to initialize the algorithm run on
the model that includes the additional PIP2 $\rightarrow$ PIP3
edge. We evaluate the likelihood function at each of the two maxima
and formulate the corresponding likelihood ratio test statistic. The
choice of $\hat{B}$ and $\hat{\Omega}$ as initial values for
estimating the larger model guarantees that the test statistics are
non-negative.

Histograms for the subsampled log-likelihood ratio statistics are
shown in Figure \ref{fig:histograms}.  The empirical distributions for
$b=30$ and $b=50$ are seen to be similar to one another and also
rather close to a $\chi^2_1$ distribution. The observed test statistic
for the full data has empirical p-value of 0.76 and 0.73 for $b=30$
and $b=50$ respectively. These p-values are slightly smaller than the
p-value of 0.78 from $\chi^2_1$ approximation.  Altogether there is
little evidence to reject the sub-model in favor of the more
complicated cyclic model.

\subsection{Simulated data}
\label{sec:simulated_data}

We now demonstrate how the BCD algorithm behaves on different types of mixed graphs.  We consider the existing R package `sem' \citep{fox:2006} as an alternative and compare the performance of these algorithms for maximum likelihood estimation on simulated data.

To simulate a mixed graph, we begin with the empty graph on $\cardY$
nodes. For $0\leq k\leq \cardY$, we add directed edges
$1\to 2 \to \cdots \to (k-1) \to k \to 1$, creating a directed cycle
of length $k$. For all $p(p-1)/2-(k-1)$ remaining pairs of nodes
$(i,j)$ with $i<j$, we generate independent uniform random variables
$U_{ij} \sim U(0,1)$. If $U_{ij} \leq d$, we introduce the directed
edge $i\to j$. Alternatively, if $d < U_{ij} \leq b+ d$, we introduce
the bi-directed edge $i\bi j$. If $U_{ij} > b+d$, there is no edge
between $i$ and $j$. After all edges have been determined, we randomly
permute the node labels. This construction ensures that the resulting
mixed graph $G$ has the following properties:
\begin{enumerate}
\item[(i)] $G$ has a unique cycle of length $k$;
\item[(ii)] $G$ is bow-free and simple.
\end{enumerate}
For this simulation, we use 24 different configurations of
$(V,N,k,d,b)$, where $N$ is the sample size.  We examine graphs of
size $\cardY=10$ and $\cardY=20$ and $N=3\cardY/2$ and $N=10\cardY$
observations. In each of these 4 configurations, we consider 3
distinct choices of the maximum cycle length: $k=0$, $\cardY/5$, and
$2\cardY/5$. For each combination of $(\cardY,N,k)$, we let $d = 0.1$
and $d = 0.2$, fixing $b=d/2$ in each case. Note that in the case of
$k=0$, every generated graph will be acyclic and simple, the class of mixed graphs considered by \citet{drton:eichler:richardson:2009}. 

\begin{table}[htbp]
\begin{center}
\begin{tabular}{cccccccccc}
\hline
\hline
\multicolumn{4}{c}{} & \multicolumn{2}{c}{Convergence} & \multicolumn{1}{c}{Both} & \multicolumn{1}{c}{Both} & \multicolumn{2}{c}{Running time} \\
\multicolumn{1}{c}{$\cardY$} & \multicolumn{1}{c}{$N$} & \multicolumn{1}{c}{$k$} & \multicolumn{1}{c}{$d$} & \multicolumn{1}{c}{BCD} & \multicolumn{1}{c}{SEM} & converge & agree & BCD & SEM\\
\hline
10 & 15 & 0 & 0.1 & 1000 & 991 & 991 & 932 & 3.8 & 24.1 \\ 
  10 & 15 & 0 & 0.2 & 1000 & 949 & 949 & 884 & 9.5 & 31.8 \\ 
  10 & 15 & 2 & 0.1 & 1000 & 479 & 479 & 456 & 10.7 & 28.7 \\ 
  10 & 15 & 2 & 0.2 & 1000 & 559 & 559 & 518 & 16.0 & 36.2 \\ 
  10 & 15 & 4 & 0.1 & 997 & 672 & 672 & 637 & 10.7 & 30.5 \\ 
  10 & 15 & 4 & 0.2 & 997 & 553 & 553 & 520 & 16.7 & 38.0 \\ 
  10 & 100 & 0 & 0.1 & 1000 & 996 & 996 & 985 & 6.5 & 30.9 \\ 
  10 & 100 & 0 & 0.2 & 1000 & 991 & 991 & 991 & 20.9 & 53.3 \\ 
  10 & 100 & 2 & 0.1 & 1000 & 517 & 517 & 517 & 40.1 & 48.0 \\ 
  10 & 100 & 2 & 0.2 & 1000 & 635 & 635 & 635 & 51.5 & 58.9 \\ 
  10 & 100 & 4 & 0.1 & 999 & 726 & 726 & 725 & 33.4 & 50.2 \\ 
  10 & 100 & 4 & 0.2 & 998 & 688 & 688 & 688 & 46.3 & 63.0 \\ 
  20 & 30 & 0 & 0.1 & 1000 & 989 & 989 & 971 & 54.0 & 324.7 \\ 
  20 & 30 & 0 & 0.2 & 1000 & 921 & 921 & 881 & 166.7 & 550.5 \\ 
  20 & 30 & 4 & 0.1 & 999 & 836 & 836 & 824 & 77.3 & 319.5 \\ 
  20 & 30 & 4 & 0.2 & 998 & 731 & 731 & 701 & 197.0 & 652.6 \\ 
  20 & 30 & 8 & 0.1 & 1000 & 709 & 709 & 696 & 97.0 & 342.2 \\ 
  20 & 30 & 8 & 0.2 & 999 & 534 & 534 & 505 & 237.5 & 766.3 \\ 
  20 & 200 & 0 & 0.1 & 1000 & 998 & 998 & 993 & 119.8 & 330.1 \\ 
  20 & 200 & 0 & 0.2 & 1000 & 983 & 983 & 958 & 299.0 & 585.4 \\ 
  20 & 200 & 4 & 0.1 & 1000 & 847 & 847 & 829 & 199.5 & 356.8 \\ 
  20 & 200 & 4 & 0.2 & 999 & 806 & 806 & 773 & 359.6 & 712.3 \\ 
  20 & 200 & 8 & 0.1 & 999 & 765 & 765 & 755 & 257.6 & 409.8 \\ 
  20 & 200 & 8 & 0.2 & 1000 & 659 & 659 & 630 & 471.7 & 851.4 \\ 
\hline
\hline
\end{tabular}
\end{center}
\caption{Data simulated from a random distribution in a randomly
  generated mixed graph model is fit to the model using BCD and the
  quasi-Newton method invoked by `sem'. Each row summarizes 1000
  simulations. `Both agree' counts the cases with ML estimates equal up to small tolerance. Running time is average CPU time (in milliseconds) for the cases in which both algorithms converged and agreed.}
\label{table:results}
\end{table}

In each simulation, we generate a random mixed graph $G$ according to
the procedure above. We then select a random distribution from the
corresponding normal model $\mathbf{N}(G)$ by taking the covariance
matrix  to be
$ \Sigma = (I-B)^{-1}\Omega(I-B)^{-T}$ for $ B\in\mathbf{B}(G)$ and
$\Omega\in\mathbf{\Omega}(G)$ selected as follows.  We set all free,
off-diagonal entries of $B$ and $\Omega$ to independent realizations
from a $\mathcal{N}(0,1)$ distribution. The diagonal entries of
$\Omega$ are chosen as one more than the sum of the absolute values of
the entries in the corresponding row of $\Omega$ plus a random draw
from a $\chi^2_1$ distribution.  Hence, $\Omega$ is diagonally
dominant and positive definite.  The model $\mathbf{N}(G)$ is then fit
to a sample of size $N$ that is generated from the selected
distribution.  We use the routine `sem' and our BCD algorithm.  The
BCD algorithm isallowed to run for a maximum of 5000 iterations, at
which point divergence was assumed.
The BCD algorithm is initialized using the procedure described in
Section \ref{sec:sachs_data}.  The `sem' method is initialized by
default using a modification of the procedure described by
\citet{mcdonald1992procedure}.

Each row of Table \ref{table:results} corresponds to 1000
simulations at a configuration of $(V,N,k,d,b)$.  In
particular, we record how often each algorithm converges.  The columns
`both converge' and `both agree' report the number of simulations for
which both algorithms converged, and the number of these simulations
for which the resulting estimates were equal up to a small tolerance.
For the routine `sem', which uses a generic `nlm' Newton optimizer, it
is not uncommon that convergence occurs but yields estimates that are
not positive definite. In these cases, we consider the algorithm to
have not converged.
 
The last two columns show the average CPU running times (in
milliseconds) over simulations for which both methods converged and
agreed\footnote{The simulations were run on a laptop with a quad-core
  2.4Ghz processor.}. We caution that these times are not directly
comparable, since `sem' computes a number of other quantities of
interest in addition to the maximum likelihood estimate.   However, the
BCD algorithm is up to 6 times faster than `sem' in some instances. 
One potential reason is that when the graph is relatively sparse, many
of the nodes may only require a single BCD update.

\section{Discussion}
\label{sec:conclusion}

This work gives is an extension of the RICF algorithm from
\citet{drton:eichler:richardson:2009} to cyclic models.  The RICF
algorithm and its BCD extension iteratively perform partial
maximizations of the likelihood function via joint updates to the
parameter matrices $B$ and $\Omega$.  Each update problem admits a
unique solution.  Like its predecessor, the generalized algorithm is
guaranteed to produce feasible positive definite covariance matrices
after every iteration. Moreover, any accumulation point of the
sequence of estimated covariance matrices is necessarily either a
local maximum or a saddle point of the likelihood function.

Despite these desirable properties, the general scope of this
algorithm to cyclic models is not without limitations. As with any
iterative maximization procedure, there is no guarantee that
convergence of the algorithm is to a global maximum, due to possible
multi-modality of the likelihood function.  In addition, for certain
models the algorithm may be ill-defined, due to collinearity of the
covariates and pseudo-covariates in our update step.  However, we show
that the models for which this occurs are non-identifiable.  Moreover,
we give necessary and sufficient graphical conditions for generically
well defined updates, which were not previously known for the acyclic
case.

In some of our simulated examples the BCD algorithm, which does not
use any overall second-order information, needed many iterations to
meet a convergence criterion.  It is possible that in those cases a
hybrid method that also consider quasi-Newton steps would converge
more quickly.  Nevertheless, our numerical experiments 
in Section \ref{sec:simulated_data} show that the BCD algorithm is
competitive in terms of computation time with the generic optimization
tools as used in the R package `sem' all the while alleviating
convergence problems.

\bibliographystyle{imsart-nameyear}
\bibliography{BCD}	

\clearpage

\begin{appendices}

\section{Proofs for claims in Section \ref{sec:sem}}
\subsection{Proof of Lemma~\ref{lem:det}}
\label{app:lem:det}

\begin{replemma}{lem:det}
  Let $B=(\beta_{ij})\in\mathbf{B}(G)$
  for a mixed graph $G$. Then
  \[
    \det{(I-B)} \ = \sum_{\sigma\in\mathbf{S}_V(G)}(-1)^{n(\sigma)}
    \prod_{i\in V(\sigma)}\beta_{\sigma(i),i}. 
  \]
\end{replemma}

\begin{proof}
  By the Leibniz formula, 
  \begin{equation}
    \det(I-B) \;=\; \sum_{\sigma\in\mathbf{S}_V} \sign(\sigma) \prod_{i\in V} (I-B)_{\sigma(i),i} \;=\; \sum_{\sigma\in\mathbf{S}_V(G)} \sign(\sigma) \prod_{i\in V} (I-B)_{\sigma(i),i}. \label{eq:SvG}
  \end{equation}
  The second equality in (\ref{eq:SvG}) holds because for all
  $\sigma\notin\mathbf{S}_V(G)$ there exists an index $i$ with
  $\sigma(i)\not=i$ and $i\to\sigma(i)\not\in E_{\to}$, which implies that
  $(I-B)_{\sigma(i),i}=-B_{\sigma(i),i}=0$ for $B\in\mathbf{B}(G)$.
  For a permutation $\sigma\in\mathbf{S}_V(G)$, and a cycle
  $\gamma \in \mathcal{C}(\sigma)$, define $V(\gamma) \subseteq V$ to be the set
  of nodes contained in the cycle $\gamma$. We may then rewrite (\ref{eq:SvG})
  as 
  \begin{align}
    \sum_{\sigma\in\mathbf{S}_V(G)} \sign(\sigma) \prod_{i\in V}
    (I-B)_{\sigma(i),i} 
    &= \sum_{\sigma\in\mathbf{S}_V(G)}
      \prod_{\gamma\in\mathcal{C}(\sigma)}\left(\sign(\gamma)
      \prod_{i\in V(\gamma)} (I-B)_{\sigma(i),i}\right)
      \notag \\ 
    &= \sum_{\sigma\in\mathbf{S}_V(G)} \prod_{\gamma\in\mathcal{C}_2(\sigma)}\left(\sign(\gamma)(-1)^{V(\gamma)} \prod_{i\in V(\gamma)} \beta_{\sigma(i),i}\right) \notag \\
    &= \sum_{\sigma\in\mathbf{S}_V(G)}(-1)^{n(\sigma)} \prod_{i\in
      V(\sigma)}\beta_{\sigma(i),i}. \label{eq:minus1_simplify}
\end{align}
The last equation (\ref{eq:minus1_simplify}) is obtained from the fact that
$\sign(\gamma)(-1)^{V(\gamma)} = -1$ for every cycle
$\gamma\in\mathcal{C}_2(\sigma)$. This follows from noting that the sign of
every even-length cycle is -1 and the sign of every odd-length cycle
is 1.
\end{proof}

\subsection{Derivation of the likelihood equations} \label{chap:lik_eqns}
Recall that the likelihood function for normal structural equation models takes the form
\begin{align*}
\ell_{G,Y}(B,\Omega) &= \frac{N}{2}\log\det\left[(I-B)^T\Omega^{-1}(I-B)\right] - \frac{N}{2}\text{tr}\left[(I-B)^T\Omega^{-1}(I-B)S\right] \\
&= \frac{N}{2}\log\det\left[(I-B)^T(I-B)\right]- \frac{N}{2}\log\det(\Omega) - \frac{N}{2}\text{tr}\left[(I-B)^T\Omega^{-1}(I-B)S\right]. 
\end{align*}
Furthermore, recall that $\beta$ and $\omega$ are the vectors of free parameters in $B$ and $\Omega$ respectively. These vectors satisfy $\text{vec}(B) = P\beta$ and $\text{vec}(\Omega)=Q\omega$. 

The first derivatives of the log-likelihood function with respect to $\beta$ and $\omega$ are
\begin{align*}
\frac{\partial \ell_{G,Y}(B,\Omega)}{\partial \beta} &= \frac{\partial P\beta}{\partial\beta} \times  \frac{\partial \ell_{G,Y}(B,\Omega)}{\partial \text{vec}(B)} \\
& = P^T \ \text{vec}\left(\frac{\partial \ell_{G,Y}(B,\Omega)}{\partial B}\right) \\
&= -\frac{N}{2}P^T \ \text{vec}\left[2(I-B)^{-T}\right.\\
  &\qquad\qquad\qquad \left.+\frac{\partial}{\partial B} \text{tr}\left(\Omega^{-1}S - B^T\Omega^{-1}S-\Omega^{-1}BS+B^T\Omega^{-1}BS\right)\right] \\
&= -\frac{N}{2}P^T \ \text{vec}\left[2(I-B)^{-T}-2\Omega^{-1}S + \frac{\partial}{\partial B} \text{tr}\left(BSB^T\Omega^{-1}\right)\right] \\
&= -\frac{N}{2} P^T \ \text{vec}\left[2(I-B)^{-T}-2\Omega^{-1}(I-B)S\right] \\
&= N P^T \ \text{vec}\left[\Omega^{-1}(I-B)S -(I-B)^{-T} \right],  \\ \\
\frac{\partial \ell_{G,Y}(B,\Omega)}{\partial \omega} &= \frac{\partial Q\omega}{\partial\omega} \times  \frac{\partial \ell_{G,Y}(B,\Omega)}{\partial \text{vec}(\Omega)} \\
& = Q^T \ \text{vec}\left(\frac{\partial \ell_{G,Y}(B,\Omega)}{\partial \Omega}\right) \\
&= -\frac{N}{2}Q^T \
  \text{vec}\left[\Omega^{-1}-\Omega^{-1}(I-B)S(I-B)^T\Omega^{-1}\right].
\end{align*} 
\qed

\section{Proofs of claims in Section \ref{sec:cycles}}
\subsection{Proof of Lemma~\ref{lem:block-update-special}}
\label{sec:proof-prop:block-update-special}

\begin{replemma}{lem:block-update-special}
In~(\ref{eq:yXac}), the vector $c$ is orthogonal to the kernel of
  $X$.
\end{replemma}
\proof The kernel of $X$ is orthogonal to the span of $X^T$.  Hence,
we have to show that
\[
c=\begin{pmatrix} 0\\ c_{i,\pa(i)} \end{pmatrix}
\;\in\;\linspan\left(\begin{pmatrix} Z_{\sib(i)}\\ Y_{\pa(i)} 
   \end{pmatrix} \right).
\]
To be clear, the $N$ columns of the displayed matrix span a
subspace of $\mathbb{R}^{\sib(i)}\times\mathbb{R}^{\pa(i)}$.  We will
in fact show something stronger, namely,
\[
c=\begin{pmatrix} 0\\ c_{i,\pa(i)}\\-c_{i,0} \end{pmatrix}
\;\in\;\linspan\left(\begin{pmatrix} Z_{\sib(i)}\\ Y_{\pa(i)} \\Y_i
   \end{pmatrix} \right).
\]

For notational convenience, let 
\begin{equation}
  \label{eq:Delta}
\Delta  \;:=\; \Omega^{-1}_{-i,-i}(I - B)_{-i}.
\end{equation}
Then 
\[
\begin{pmatrix} Z_{\sib(i)}\\ Y_{\pa(i)} \\Y_i
   \end{pmatrix} = 
\begin{pmatrix} 
 0& \Delta_{ \sib(i) , \pa(i)} & \Delta_{ \sib(i)
    ,V \setminus (\pa(i)\cup\{i\})}\\
 0&I_{\pa(i)} & 0\\
  1&0&0
    \end{pmatrix} 
\begin{pmatrix}
  Y_i\\
  Y_{\pa(i)} \\ Y_{V\setminus(\pa(i)\cup\{i\})}
\end{pmatrix}.
\]
Since the data matrix $Y$ is assumed to have full rank, it suffices to
show that there is a vector $w\in\mathbb{R}^V$ such that
\begin{equation}
  \label{eq:a-ortho-X-to-show}
  \begin{pmatrix}
    0&\Delta_{ \sib(i) , \pa(i)} & \Delta_{ \sib(i)
    , V \setminus  (\pa(i)\cup\{i\})} \\
  0&I_{\pa(i)} & 0\\
  1&0&0
   \end{pmatrix}  w \;=\;\begin{pmatrix} 0\\ c_{i,\pa(i)}\\-c_{i,0} \end{pmatrix} .
\end{equation}

Consider any node $p \in \pa(i)$.  For a permutation
$\sigma\in\mathbf{S}_V$, i.e., a permutation of the vertex set $V$,
let $\gamma_i(\sigma)$ be the permutation cycle containing $i$.
Then, from Lemma \ref{lem:det}, the vector $c$ has the coordinate
indexed by $p$ equal to
\begin{align}
\notag
c_p &\;=\; \sum_{\substack{\sigma \in \textbf{S}_V(G) \\ \sigma:\sigma(p) = i}} (-1)^{n(\sigma)} \prod_{j \in V(\gamma_i(\sigma))\setminus \{p\}} \beta_{\sigma(j),j} \prod_{\gamma \in \mathcal{C}_2(\sigma)\setminus\{\gamma_i(\sigma)\}}\prod_{k \in V(\gamma)} \beta_{\sigma(k),k}\,.
\end{align}
Let $\Psi_{i}$ be the set of all directed cycles in the graph $G$ that
contain node $i$.  For later convenience, we also include in $\Psi_i$
a self-loop $i\to i$.  If $\gamma\in\Psi_i$, then write
$P^{\gamma}_{i:k}$ for the product of coefficients $\beta_{lj}$ for
edges $j\to l$ that lie on the directed path from $i$ to $k$ that is
part of cycle $\gamma$, with $P^{\gamma}_{i:i}=1$.  We set
$P^{\gamma}_{i:k}=0$ if $k\not \in \gamma$.  Then
\begin{align}
\label{eq:def-xi}
c_p &\;=\; \sum_{\gamma \in \Psi_{i} } P^{\gamma}_{i:p} \left(\sum_{\substack{\sigma \in \textbf{S}_V(G) \\ \sigma:\gamma \subseteq \sigma}} (-1)^{n(\sigma)} \prod_{\substack{\gamma' \in \mathcal{C}_2(\sigma)\\
 \gamma' \neq  \gamma }}\prod_{j \in V(\gamma')}
\beta_{\sigma(j),j}\right) \;=:\; \sum_{\gamma \in \Psi_{i} } P^{\gamma}_{i:p} \xi_{\gamma}.
\end{align}

Similarly, we have 
\[
c_{i,0} \;=\; -\sum_{\substack{\sigma \in \textbf{S}_V(G) \\
    \sigma:\sigma(i) = i}} (-1)^{n(\sigma)}  \prod_{\gamma \in
  \mathcal{C}_2(\sigma)\setminus\{(i)\}}\prod_{k \in V(\gamma)}
\beta_{\sigma(k),k}\,.
\]
The negative sign in front is due to the fact that in the matrix $I-B$
the off-diagonal entries are negated but the diagonal entries are not.
Having included the self-loop $i\to i$ in $\Psi_i$, we obtain that
\begin{equation}
\label{eq:-c0}
  c_{i,0} \;=\; - \sum_{\gamma \in \Psi_{i} } P^{\gamma}_{i:i}
  \xi_{\gamma}
\;=\; - \sum_{\gamma \in \Psi_{i} } \xi_{\gamma}.
\end{equation}

  Using the sums $\xi_\gamma$ from~(\ref{eq:def-xi}),
define $w$ to be the vector with coordinates
\begin{equation}
  \label{eq:w}
w_k = \sum_{\gamma \in \Psi_{i}} P^{\gamma}_{i:k} \xi_\gamma,  \quad k\in V.
\end{equation} 
By~(\ref{eq:def-xi}) and~(\ref{eq:-c0}), the vector $w$ satisfies the
equations in~(\ref{eq:a-ortho-X-to-show}) that are indexed by
$\pa(i)\cup\{i\}$.  We now show that it also satisfies those indexed
by $\sib(i)$.

Let $k\in V\setminus\{i\}$.  Since every directed path from $i$ to $k$
passes through one of the parents of $k$, we have
\[
P^\gamma_{i:k} \;=\; \sum_{p\in\pa(k)} P^\gamma_{i:p}\beta_{kp}
\]
for any cycle $\gamma\in\Psi_i$ that contains $k$.  Therefore,
\[
\left((I - B)_{-i} w \right)_k = w_k - \sum_{p \in \pa(k)} \beta_{kp}w_p
= \sum_{\gamma\in \Psi_i} P^\gamma_{i:k}\xi_\gamma - \sum_{p \in \pa(k)}
\beta_{kp}\sum_{\gamma\in \Psi_i} P^\gamma_{i:p}\xi_\gamma
=0.
\]
In other words, $w$ is in the kernel of $(I - B)_{-i}$.  This kernel
is contained in the kernel of $\Delta_{\sib}$, which yields~(\ref{eq:a-ortho-X-to-show}).
\qed

\subsection{Proof of Lemma \ref{lem:optimum}}
\label{app:ratio-lemma}

 \begin{replemma}{lem:optimum}
 
  For constants $a,b,c_0,c_1\in\mathbb{R}$
    with $c_1\neq0$, define the function
    \[
      f(x) = \frac{(a-x)^2+b^2}{(c_0+c_1x)^2}, \qquad
      x\in\mathbb{R}\setminus\left\{-c_0/c_1\right\}. 
    \]
    \begin{enumerate}[label=(\roman{*}), ref=(\roman{*}),leftmargin=5.0em]
    \item If $c_0+ac_1 \neq 0$, then $f$ is uniquely minimized by
      \[
        x = \frac{ac_0+a^2c_1+b^2c_1}{c_0+ac_1} = a +
        \frac{b^2c_1}{c_0+ac_1}.
      \]
    \item   If $c_0+ac_1=0$ and
      $b=0$, then $f$ is constant and equal to $1/c_1^2$.
    \item 
      If $c_0+ac_1 = 0$ and $b^2>0$, then $f$ does not achieve its
      minimum, and
      $\inf f = \lim_{x\to\pm\infty} f(x) = 1/c_1^2$.
    \end{enumerate}
  \end{replemma}

\proof
The lemma is concerned with the univariate rational function
\[
f(x) = \frac{(a-x)^2+b^2}{(c_{i,0}+c_1x)^2}, \qquad x\in\mathbb{R},
\]
where $a,b,c_{i,0},c_1\in\mathbb{R}$ are constants with $c_1\neq0$.  The
function $f$ is defined everywhere except at the point $x=-c_{i,0}/c_1$. The
limits of $f$ are
\begin{equation} \label{eq:limits}
\lim_{x\to\infty} f(x) = \lim_{x\to-\infty} f(x) = \frac{1}{c_1^2}.
\end{equation}
Note that
\begin{equation} \label{eq:first_derivative}
f'(x) = - \frac{2(ac_{i,0}+a^2c_1+b^2c_1-c_{i,0}x-ac_1x)}{(c_{i,0}+c_1x)^3}.
\end{equation}
Equating (\ref{eq:first_derivative}) to zero and solving results in one critical point:
\begin{equation} \label{eq:crit_pt}
x^\star = \frac{ac_{i,0}+a^2c_1+b^2c_1}{c_{i,0}+ac_1},
\end{equation}
which is finite if $c_{i,0}+ac_1 \neq 0$. Moreover, observe that
\begin{equation} \label{eq:crit_val}
f(x^\star) = \frac{b^2}{(c_{i,0}+ac_1)^2+b^2c_1^2} < \frac{1}{c_1^2}.
\end{equation}
The second derivative of $f$ is given by
\[
f''(x) = \frac{2(c_{i,0}^2+4ac_{i,0}c_1+3a^2c_1^2+3b^2c_1^2-2c_{i,0}c_1x-2ac_1^2x)}{(c_{i,0}+c_1x)^4},
\]
from which it is revealed that
\[
f''(x^\star) = \frac{2(c_{i,0}+ac_1)^4}{\left((c_{i,0}+ac_1)^2+b^2c_1^2\right)^3} >0.
\]
Hence, we see that $x^\star$ in (\ref{eq:crit_pt}) is the unique critical point and is a local minimum. Moreover, from (\ref{eq:limits}) and (\ref{eq:crit_val}), we see that $x^\star$ must be the global minimum. 
If instead $c_{i,0}+ac_1 = 0$ and $b\not=0$, then (\ref{eq:crit_pt}) reveals
that there are no critical points, and at the pole we have
$\lim_{x\to -c_{i,0}/c_1}f(x) = \infty$. It thus follows that $f$ achieves its
minimum at $x = \pm \infty$. The case of  $c_{i,0}+ac_1=0$ and $b=0$ is clear. \qed

\begin{figure}[t]
\centering
\includegraphics[scale = .75]{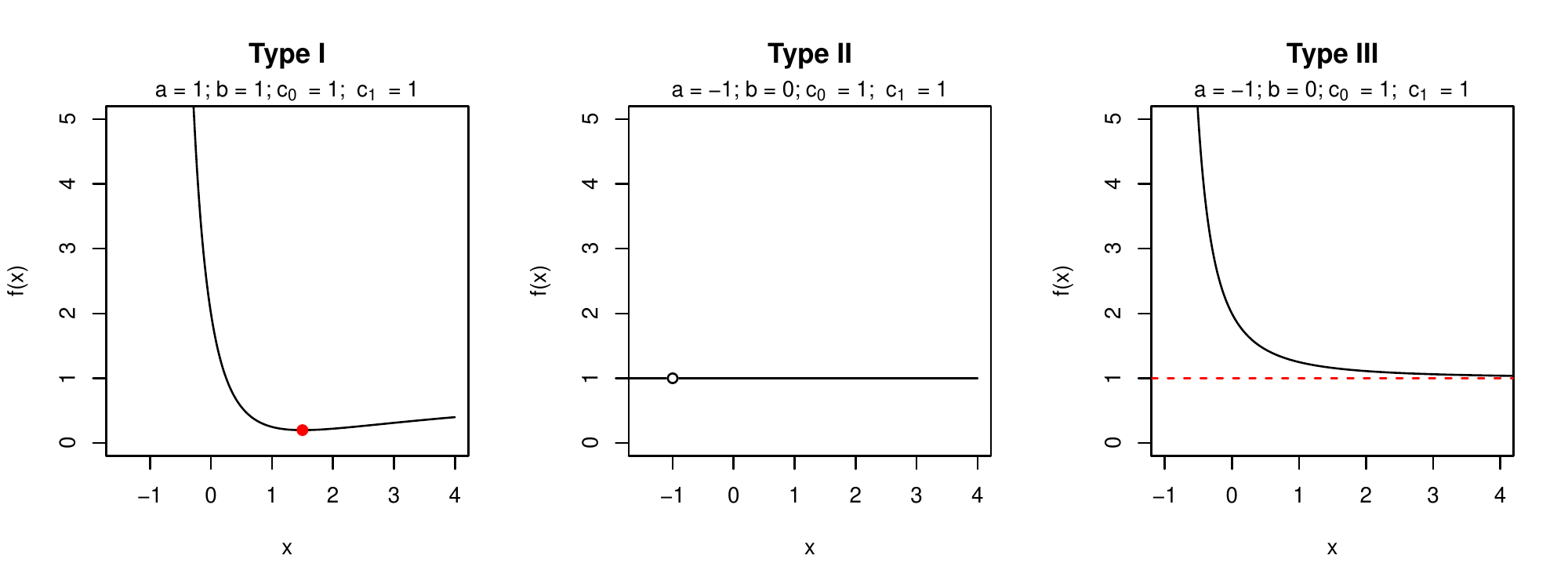}
\caption{Examples of functions from the three cases in Lemma
  \ref{lem:optimum}. A unique minimum is achieved only in Type I
  (indicated by the point).}
\end{figure}

\subsection{Proof of claim in Remark~\ref{rem:y02zeroA2fails}}
\label{app:proof:rem:y02zeroA2fails}

We use the notation from the proof of
Lemma~\ref{lem:block-update-special}.  Since $Y_i^T$ is in the span of
$X$, we have $y_0^2=0$ and the least squares vector $\hat\alpha$
satisfies that
  \[
    Y_i = e_i^T Y = \hat\alpha^T \begin{pmatrix} Z_{\sib(i)}\\ Y_{\pa(i)} 
    \end{pmatrix} = \hat\alpha^T
    \begin{pmatrix} 
      \Delta_{ \sib(i) , \pa(i)} & \Delta_{ \sib(i)
        , V \setminus \pa(i)}\\
      I_{\pa(i)} & 0
    \end{pmatrix} 
    \begin{pmatrix}
      Y_{\pa(i)} \\ Y_{V\setminus\pa(i)}
    \end{pmatrix},
  \]
  where $e_i$ is the $i$-th canonical basis vector.
  Since $Y$ has full rank, it follows that 
  \[
   e_i^T \;=\;  \hat\alpha^T \begin{pmatrix} 
      \Delta_{ \sib(i) , \pa(i)} & \Delta_{ \sib(i)
        , V \setminus \pa(i)}\\
      I_{\pa(i)} & 0
    \end{pmatrix}.
  \]
  In the proof of Lemma~\ref{lem:block-update-special}, we constructed
  a vector $w\in\mathbb{R}^V$ such that 
  \[
    c=
    \begin{pmatrix}
      0 \\ c_{i,\pa(i)}\\ -c_{i,0}
    \end{pmatrix}
    =\begin{pmatrix} 
      0&\Delta_{ \sib(i) , \pa(i)} & \Delta_{ \sib(i)
        , V \setminus \pa(i)}\\
      0&I_{\pa(i)} & 0\\
      1&0&0
    \end{pmatrix} w.
  \]
  We find that
  \[
    c^T\hat\alpha^T = \langle \begin{pmatrix} 
      \Delta_{ \sib(i) , \pa(i)} & \Delta_{ \sib(i)
        , V \setminus \pa(i)}\\
      I_{\pa(i)} & 0
    \end{pmatrix}w, \hat\alpha\rangle = \langle w, e_i\rangle = w_i = -c_{i,0}.
  \]

\section{Proofs of claims in Section \ref{sec:assumptionDiscussion}}
\subsection{Proof of Proposition \ref{lem:fullRankX}}
\label{sec:fullRankXProof}

\begin{repproposition}{lem:fullRankX}
  Let $G=(V,E_{\to},E_{\bi})$
  be a mixed graph, and let $i\in V$.  Then the following two
  statements are equivalent:
  \begin{enumerate}
  \item[(a)] Condition~\ref{A1} holds for generic triples
    $(Y,B,\Omega)\in\mathbb{R}^{V\times
      N}\times\mathbf{B}(G)\times\mathbf{\Omega}(G)$.
  \item[(b)] The induced subgraph $G_{-i}$ contains a system of
    half-collider paths from a subset of
    $V \setminus (\pa(i)\cup\{i\})$ to $\sib(i)$ such that the
    bi-directed portions are pairwise disjoint.
  \end{enumerate}
\end{repproposition}

\proof

For notational convenience, let
\[
\Lambda \;:=\; (I - B)_{-i} \qquad\text{and}\qquad
\Delta  \;:=\; \Omega^{-1}_{-i,-i}\Lambda   = \Omega^{-1}_{-i,-i}(I -
B)_{-i} .
\]

Then 
\begin{equation}
  \label{eq:ZsibYpaii}
\begin{pmatrix} Z_{\sib(i)}\\ Y_{\pa(i)} \\Y_i
   \end{pmatrix} = 
\begin{pmatrix} 
 0& \Delta_{ \sib(i) , \pa(i)} & \Delta_{ \sib(i)
    ,V \setminus (\pa(i)\cup\{i\})}\\
 0&I_{\pa(i)} & 0\\
  1&0&0
    \end{pmatrix} 
\begin{pmatrix}
  Y_i\\
  Y_{\pa(i)} \\ Y_{V\setminus(\pa(i)\cup\{i\})}
\end{pmatrix}.
\end{equation}

$(a)\Longrightarrow(b)$. If the matrix above has full row rank, then
$rk\left(\Delta_{ \sib(i) , V\setminus \{\pa(i)\cup i\}}\right) =
|\sib(i)| = S_i$. This implies that there exists a $S_i \times S_i$
submatrix of
$\Delta_{ \sib(i) , V\setminus \left(\pa(i)\cup\{i\}\right)}$ which
has full rank. Let that full rank sub-matrix be
$\Delta_{\sib(i) , \tilde V}$ where
$\tilde V \subseteq V\setminus \{\pa(i) \cup i\}$ where
$|\tilde V| = S_i$.  By the Cauchy-Binet formula, we have
\begin{equation*}
\det \Delta_{\sib(i) , \tilde V} = \sum_{A \in \left\{V \setminus i \atop S_i\right\}} \det\left[\Omega^{-1}_{-i,-i}\right]_{\sib(i) , A} \det \Lambda_{A , \tilde V}. 
\end{equation*}
Since $\det\Delta_{\sib(i) , \tilde V} \neq 0$ by construction, there
must exist a set $A \subseteq V \setminus \{i\}$ with $|A| = S_i$ for
which both $\det[\Omega^{-1}_{-i,-i}]_{\sib(i) , A} \neq 0$ and
$\det\Lambda_{A , \tilde V} \neq 0 $. 

Let $D$ be a $\mathbb{R}^{V-1 \times V-1}$ diagonal matrix with
$\left(\sqrt{\omega_{11}}, \ldots \sqrt{\omega_{i-1,i-1}},
  \sqrt{\omega_{i+1,i+1}} \ldots \sqrt{\omega_{VV}}\right)$ on the
diagonal.  Let $I - W$ be the correlation matrix corresponding to
$\Omega_{-i,-i}$, so that $D(I-W)D = \Omega_{-i,-i}$ and $(I-W)^{-1} =
D\Omega^{-1}_{-i,-i}D$. Then, $\det[\Omega^{-1}_{-i,-i}]_{\sib(i) , A}
\neq 0$ implies that
\begin{align*}
0 &\neq \det D_{\sib(i) , \sib(i)} \det[\Omega^{-1}_{-i,-i}]_{\sib(i) , A} \det D_{A , A}\\
 &=  \sum_{B,C \in {V\setminus i \choose S_i}}\det D_{\sib(i) , B} \det[\Omega^{-1}_{-i,-i}]_{B , C} \det D_{C , A}\\
 & = \det \left[D\Omega^{-1}_{-i, -i}D\right]_{\sib(i) , A}\\
 & = \det \left[(I-W)^{-1}\right]_{\sib(i), A},
\end{align*}
where the first equality holds because $\det D_{J,K} \neq 0$ iff
$J = K$ since $D$ is diagonal.  Applying Corollary 3.8 from
\citet{sullivant:2010}, we see that
$\det \left[(I-W)^{-1}\right]_{\sib(i), A} \neq 0$ implies that there
exists a system $\mathcal{P}$ of vertex disjoint bi-directed paths
from $\sib(i)$ to $A$.

A Leibniz expansion of $\det\Lambda_{A , \tilde V}$ shows that
$\det\Lambda_{A , \tilde V} \not=0$ implies that
the graph contains a matching of $\tilde V$ and $A$.  In other words,
we can enumerate the sets as $\tilde V=\{v_1,\dots,v_{S_i}\}$ and
$A=\{a_1,\dots,a_{S_i}\}$ such that $v_s=a_s$ or $v_s\in\pa(a_s)$ for
$s=1,\dots, S_i$.  Since we have the system $\mathcal{P}$ of vertex
disjoint bi-directed paths from $\sib(i)$ to $A$, there is thus a
system of half-collider paths with pairwise disjoint bi-directed
portions.  The paths are fully contained in $\mathcal{G}_{-i}$ because
we considered $\Omega_{-i,-i}$ and
$\tilde V \subseteq V\setminus \left(\pa(i) \cup \{i\}\right)$.

$(a) \Longleftarrow(b)$. When $Y$ is full rank (which is true for
generic $Y$ when $N \geq |V|$), a drop in the row rank of the matrix
displayed in~(\ref{eq:ZsibYpaii}) is equivalent to a drop in rank of
$\Delta_{ \sib(i), V \setminus \left(\pa(i)\cup \{i\}\right)}$.  We
show that given a set of nodes $\tilde V$
that satisfies the assumed graphical condition, the matrix
$\Delta_{ \sib(i), V \setminus \left(\pa(i) \cup \{i\} \right)}$ is
generically of full rank since there is an $S_i \times S_i$ minor
which only vanishes on a set of pairs $(B,\Omega)$ with Lebesgue
measure 0.  Here, $|\tilde V|=|\sib(i)|=S_i$.  In what follows we consider systems of half-collider paths
from $\tilde V$ to $\sib(i)$.  We always index the paths as $\pi^s$
where $s$ is the endpoint in $\sib(i)$.  We write $v_s$ for the other
endpoint of $\pi^s$, so $v_s\in\tilde V$.

First, given a valid half-collider path system
$\mathcal{P} = \{\pi^s\}_{s\in \sib(i)}$, we claim that there exists
an ordering $\prec$ of $\sib(i)$ and a system of half-collider paths
$\mathcal{\hat P}=\{\hat\pi^s\}_{s\in \sib(i)}$ such that $r \prec s$
implies that the first node of $\hat \pi^s$ is not
in the bi-directed portion of $\hat \pi^r$. Note that if the first
node of a half-collider path is the tail of a directed edge, it can
still appear in the bi-directed portion of another path in a valid
path system. Suppose that the specified ordering does not exist.  Then
there is a set of nodes
$G_{\text{cycle}} = \{g_1, g_2\dots g_q\} \subseteq \sib(i)$ where the
first node of $\pi^{g_k}$ is in the bi-directed portion of
$\pi^{g_{k-1}}$ for $k > 1$ and the first node of $\pi^{g_1}$ is in
the bi-directed portion of $\pi^{g_q}$. Then this implies that for
each $g_i \in G_\text{cycle}$, there exists some node in the
half-collider path $\pi^{g_i}$ closer to $g_i$ which is not in
$\pa(i)$, namely the first node of $\pi^{g-1}$. Thus, simply removing
the first directed edge of each of the half-collider paths produces a
valid system $\mathcal{\hat P}$ that can be ordered as claimed.  For
the remainder of the proof, we assume without loss of generality that
the considered system of half-collider paths $\mathcal{P}$ has the
desired ordering.

Consider the $S_i \times S_i$ sub-matrix
$\Delta_{ \sib(i), \tilde V}$, and note that the determinant of
$\Delta_{\sib(i), \tilde V}$ is a rational function of the elements of
$\Omega$ and $B$.  To show that both the numerator and denominator
only vanish on a null set, we appeal to Lemma 1 from
\citet{okamoto:1973} which states that the zero set of a nonzero
polynomial is a null set.  To show that the numerator and denominator
are not 0 everywhere, consider the  point
$\left(\Omega^\dagger, B^\dagger \right)$ whose coordinates are specified
as follows:
\begin{itemize}
\item Set all diagonal entries $\omega_{kk}^\dagger$ to 1;
\item Set $\omega_{jk}^\dagger > 0$ (but sufficiently small so that $\Omega$ is positive definite) if and only if $j \leftrightarrow k \in  \pi^s$ for some $s \in \sib(i)$;
\item Set $\beta_{k v_s}^\dagger = -1$ if and only if $ k \leftarrow v_s \in \pi_s$;
\item Set all other parameters to 0.
\end{itemize}
  Note that the
support of $\left(\Omega^\dagger, B^\dagger \right)$ matches the edges
of the half-collider paths in $\mathcal{P}$.
Let $\delta^\dagger_{jk}$ and $\lambda^\dagger_{jk}$ be the entries of the matrices
$\Delta_{ \sib(i), \tilde V}$ and $\Lambda$ constructed from
$\left(\Omega^\dagger, B^\dagger \right)$.   Let $\omega^\star_{jk}$ be
the entries of $(\Omega^\dagger_{-i,-i})^{-1}$.  Then
\begin{equation*}
0 \neq \tilde \delta^\dagger_{sv_s} = \sum_{k \in V \setminus i}\omega^\star_{sk} \lambda^\dagger_{kv_s} = \begin{cases} 
\omega^\star_{ss} &\mbox{ if } v_s = s, \\
\omega^\star_{sm}  &\mbox{ if the first edge in $ \pi^s$ is $v_s \rightarrow m$}, \\
 \omega^\star_{s v_s} &\mbox{ if the first edge in $\pi^s$ is bi-directed} .
\end{cases}
\end{equation*}
For any $r\neq s$, 
\begin{equation*}
\delta^\dagger_{s v_r} = \sum_{k \in V \setminus i}\omega^\star_{sk} \lambda^\dagger_{k v_s} = 
\begin{cases} \omega^\star_{s a_r} & \mbox{ if the first node of $ \pi^r$ is in $\pi^s$},\\
0 &\mbox{ else} .
\end{cases}
\end{equation*} 
By the assumed ordering of $\sib(i)$, if the first node of $\pi_r$
lies on $\pi_s$ then $s \succ r$.  Therefore, there is a permutation
of the rows and columns of $\Delta_{ \sib(i), \tilde V}$ that makes
the matrix upper triangular with $\delta^\dagger_{s\hat v_s}$ on the
diagonal.  Hence, the determinant is
\begin{equation*}
\det \Delta^\dagger_{\sib(i) , \tilde V} = \prod_s\tilde \delta^\dagger_{s \hat v_s}  \neq 0.
\end{equation*}
It follows that the determinant of $ \Delta_{\sib(i) , \tilde V}$ is
nonzero almost everywhere.  The assumed graphical condition thus
implies that the matrix in~(\ref{eq:ZsibYpaii}) has generically
full rank. \qed

\subsection{Proof of Theorem \ref{thm:unique_updates}}
\label{sec:unique_updatesProof}

\begin{reptheorem}{thm:unique_updates}
For a mixed graph
  $G=(V,E_{\to},E_{\bi})$, the following two statements are equivalent:
  \begin{enumerate}
  \item[(a)] For all $i\in V$, the induced subgraph $G_{-i}$ contains
    a system of half-collider paths from a subset of
    $V \setminus (\pa(i)\cup\{i\})$ to $\sib(i)$ such that the
    bi-directed portions are pairwise disjoint.
  \item[(b)] For generic triples
    $(Y,B_0,\Omega_0)\in\mathbb{R}^{V\times
      N}\times\mathbf{B}(G)\times\mathbf{\Omega}(G)$, any finite
    number of iterations of the BCD algorithm for $\ell_{G,Y}$ have
    unique and feasible block updates when 
    $(B_0,\Omega_0)$ is used as starting value.
  \end{enumerate}
\end{reptheorem}
\begin{proof}
  If the graphical condition (a) fails then, by
  Proposition~\ref{lem:fullRankX} and Theorem~\ref{thm:A1A2}, there
  exists a node $i\in V$ at which the BCD algorithm does not have a
  unique and feasible update, irrespective of the choice of
  $(Y,B_0,\Omega_0)$.

  Conversely, suppose condition (a) holds.  Let $i(t)$ be the node
  considered in step $t$ of the BCD algorithm, and let
  $(B_t,\Omega_t)$ be the pair of parameter matrices after the $t$-th
  block update.  
  By Theorem~\ref{thm:A1A2}, each pair $(B_t,\Omega_t)$ is a rational
  function of the input triple $(Y,B_{t-1},\Omega_{t-1})$.  For
  $i=i(t)$, the conditions~\ref{A1} and~\ref{A2} from
  Proposition~\ref{lem:fullRankX} are rational conditions on
  $(Y,B_{t-1},\Omega_{t-1})$.  Therefore, for any $T\ge 1$ there is a
  rational function $f_T(Y,B,\Omega)$ such that
  $f_T(Y,B_0,\Omega_0)\not=0$ if and only if the BCD updates at
  $i(1),\dots,i(T)$ all have unique feasible solutions.  By
  \citet[Lemma 1]{okamoto:1973}, it now suffices to show that there exists
  a triple $(Y,B_0,\Omega_0)$ such that when started at
  $(B_0,\Omega_0)$ the BCD updates at $i(t)$, $t\le T$, all have
  unique and feasible solutions.

  When condition (a) holds and $Y$ is full rank,~\ref{A1} holds for
  generic $\Omega$ and $B$ for all $i\in V$.  Thus, we may pick
  $\Omega_0$ and $B_0$ such that~\ref{A1} holds for every $i \in V$.
  Now choose $Y$ as in the proof of Proposition~\ref{prop:A2generic},
  with $(I-B_0)^{-1}\Omega_0(I-B_0)^{-T}=\frac{1}{N} YY^T$.  Then as
  shown when proving Proposition~\ref{prop:A2generic},
  condition~\ref{A2} holds for every $i\in V$.  By
  Proposition~\ref{lem:fullRankX}, the first BCD update problem has a
  unique feasible solution.  This solution is
  $(B_1,\Omega_1)=(B_0,\Omega_0)$ because $(B_0,\Omega_0)$ is a global
  maximizer of the likelihood function by definition of $Y$.  By
  induction, we have $(B_t,\Omega_t)=(B_0,\Omega_0)$ at all steps $t$.
  Consequently, for the triple $(Y,B_0,\Omega_0)$, any finite number
  of BCD updates have unique and feasible solutions, as we needed to
  show.
\end{proof}

\subsection{Verifying graphical condition in polynomial time} 
\label{sec:poly-time}

\begin{repproposition}{prop:poly-time}
  For any mixed graph $G=(V,E_{\to},E_{\bi})$, condition (a) in
  Theorem~\ref{thm:unique_updates} can be checked in
  $\mathcal{O}(|V|^5)$ operations.
\end{repproposition}

\begin{proof}
  In order to show that condition (a) from Theorem
  \ref{thm:unique_updates} can be checked in $\mathcal{O}(|V|^5)$
  operations, it suffices to show that the condition imposed at each
  node $i\in V$ can be checked in $\mathcal{O}(|V|^4)$ operations.  So
  fix an arbitrary node $i\in V$.

  If $\pa(i) \bigcap \sib(i) = \emptyset$, the condition is trivially
  satisfied by taking the set $\tilde V = \sib(i)$.  If
  $\pa(i) \bigcap \sib(i) \not= \emptyset$ (i.e.,
  $j \rightarrow i \in E_\rightarrow$ and
  $j\leftrightarrow i\in E_\leftrightarrow$ for some $j \in V$), we
  can check the condition by considering a relevant flow network $\hat G_i$
  which captures all half-collider paths to $\sib(i)$.  In this
  network, we include a sink node connected to each node in $\sib(i)$
  and a source node connected to each ``allowable'' node, that is, any
  node $j \not \in \pa(i)$ that can reach $\sib(i)$ through a
  half-collider path. In this network $\hat G_i$, the half-collider
  path criterion of Theorem $\ref{thm:unique_updates}$ is satisfied if
  and only if the maximum flow from the source to the sink is equal to
  $|\sib(i)|$.

  More specifically, the flow network $\hat G_i$ is constructed as
  follows:
  \begin{enumerate}
  \item[(i)] Include a source $q$ and a sink $t$.
  \item[(ii)] Let $C\subseteq V \setminus\{i\}$ be the set of all nodes
    with a bi-directed path to $i$ in ${G}$.
    \begin{enumerate}
    \item Add a node $b(j)$ to $\hat G_i$ for each $j \in C$.
    \item For all $j,k \in C$ with
      $j \leftrightarrow k \in E_\leftrightarrow$, add edges
      $b(j) \rightarrow b(k)$ and $b(k) \rightarrow b(j)$ to
      $\hat G_i$.
    \item For each $j \in C \setminus \pa(i)$, if $j \in \pa(C)$,
      then include node $q(j)$ and add edges
      $q \rightarrow q(j) \rightarrow b(j)$ to $\hat G_i$.  If
      $j\not\in\pa(C)$, then add edge $q \rightarrow b(j)$ to
      $\hat G_i$.
    \item For each
      $j \in \pa(C) \setminus \left(\pa(i)\cup \{i\}\right)$, add
      node $d(j)$ and edges $d(j) \rightarrow b(m)$ for all
      $m \in C$ with $j\in\pa(m)$.  If $j \in C $ then add edge
      $q(j) \rightarrow d(j)$, and if $j\not\in C$, add edge
      $q \rightarrow d(j)$.
    \item For each $j \in \sib(i) \cap C$, add edge $b(j) \rightarrow t$.
    \end{enumerate}
  \item[(iii)] Let the capacity of $q$ and $t$ be $|\sib(i)|$. Let
    all other node capacities be 1.  Set all edge capacities to 1.
  \end{enumerate}
  The network $\hat G_i$ includes a directed path from source $q$ to
  sink $t$ to represent each valid half-collider path from
  $V \setminus \left(\pa(i)\cup \{i\}\right)$ to $s \in \sib(i)$.  In
  doing so, it is important to keep track of whether a nodes is in the
  bi-directed part of one half-collider path and the directed part of
  another half-collider path.  To capture the directed and bi-directed
  roles of a node, respectively, we represent any node
  $j \in C \cap \pa(C) \setminus \left(\pa(i)\cup \{i\}\right)$ with
  the two nodes $d(j)$ and $b(j)$ in $\hat G_i$.  However, in order to
  ensure that each node $j$ is the beginning node for only one
  half-collider path in the path system (so that $\tilde V$ has
  cardinality $|\sib(i)|$), a bottleneck node $q(j)$ is included to
  ensure that at most 1 total unit of flow ``originates" at the
  representations of $j$.

  As shown in Lemma~\ref{thm:checkCondition} below, the graphical
  condition can be checked by solving the maximum flow problem for
  $\hat G_i$.  The standard max flow problem (with edge capacities but
  not node capacities) can be solved in
  $\mathcal{O}(|\hat V|^2 |\hat E|)$, where $\hat V$ and $\hat E$
  are the vertex and edge set of the network, respectively; see
  \cite{edmonds:1970}.  To encode the node capacity constraints, we
  augment our network so that each node has an additional in/out node
  with a single edge pointing to/from the original node with the
  original node capacity.  In the constructed network
  $|\hat V | \leq 5|V|$ and $|\hat E| \leq |\hat V|^2$ so the max flow
  problem for each individual node is $\mathcal{O}(|V|^4)$, as was our
  claim.
\end{proof}

\begin{lemma} \label{thm:checkCondition} In the given mixed graph $G$,
  there exists a set of $|\sib(i)|$ nodes
  $\tilde V \subseteq V \setminus (\pa(i)\cup\{ i\})$ such that there
  is a system of half-collider paths $\mathcal{P}$ from $\tilde V$ to
  $\sib(i)$ where the bi-directed components are vertex distinct and
  do not include $i$, if and only if the constructed network
  $\hat G_i$ has maximum flow from source $q$ to sink $t$ of $|\sib(i)|$.
\end{lemma}

\proof Suppose there exists a system of half-collider paths
$\mathcal{P}$ that satisfies the given condition.  Then there is a
corresponding system of paths from source $q$ to source $t$ in
$\hat G_i$ with flow 1 over each edge.  Since the bi-directed
components are vertex distinct and each node in $\tilde V$ is
distinct, we do not use any of the nodes in $\hat G_i$ more than once
so none of the node or edge capacities are exceeded.  Since there are
$|\sib(i)|$ half-collider paths, the total flow is also $|\sib(i)|$.

Conversely, suppose the maximum flow on $\hat G_i$ is $|\sib(i)|$.
Because each capacity is integer valued, the flow can be decomposed
into a system $\mathcal{P}$ of directed paths with integer flow
\citep{ford1962flows}.  By construction, each path in $\mathcal{P}$ is
a valid half-collider path which begins at some node
$j \not \in \left(\pa(i)\cup \{i\}\right)$ and ends at some
$s \in \sib(i)$. In addition, since $q(j)$ only has capacity 1, the
edges $q(j) \rightarrow d(j)$ and $q(j) \rightarrow b(j)$ cannot
simultaneously be utilized in $\mathcal{P}$.  Thus, each edge which
carries flow from the source $q$ represents a distinct node in
$V\setminus \left(\pa(i)\cup \{i\}\right)$ and
$|\tilde V| = |\sib(i)|$. Because there are only $|\sib(i)|$ edges to
the sink, each connected to some node representing $s \in \sib(i)$,
there is a half-collider path to each $s \in \sib(i)$. Finally, since
the capacity of each node in $\hat G_i$ is one, each node $b(j)$ only
appears once in the system which makes the bi-directed portions
vertex distinct. \qed

\end{appendices}

\end{document}